\documentclass[11pt]{article}

\usepackage[active]{srcltx}
\usepackage{hyperref}
\usepackage{amsmath,amssymb}
\usepackage[all]{xy}
\usepackage{graphicx}
\usepackage{color}
 1 
\newtheorem{theorem}{Theorem}
\newtheorem{prop}{Proposition}
\newtheorem{corol}{Corollary}
\newtheorem{lemma}{Lemma}
\newtheorem{definition}{Definition}

\def\qed{\ifvmode\Realemovelastskip\fi
{\unskip\nobreak\hfil\penalty50\hbox{}\nobreak\hfil \hbox{\vrule
height1.2ex width1.2ex}\parfillskip=0pt \finalhyphendemerits=0
\par\smallskip}}
\def\qedr{\ifvmode\Realemovelastskip\fi
{\unskip\nobreak\hfil\penalty50\hbox{}\nobreak\hfil \hbox{
$\diamond$}\parfillskip=0pt \finalhyphendemerits=0
\par\smallskip}}
\def\ds{\displaystyle}
\newenvironment{proof}{\noindent{\sl Proof:~~~}}{\quad \qed}

\def\beq{\begin{equation}}
\def\eeq{\end{equation}}
\def\bea{\begin{eqnarray}}
\def\eea{\end{eqnarray}}
\def\beann{\begin{eqnarray*}}
\def\eeann{\end{eqnarray*}}
\def\beasn{\begin{sneqnarray}}
\def\eeasn{\end{sneqnarray}}
\def\ben{\begin{enumerate}}
\def\een{\end{enumerate}}
\def\bit{\begin{itemize}}
\def\eit{\end{itemize}}

\def\derpar#1#2{\displaystyle\frac{\partial{#1}}{\partial{#2}}}
\def\derpars#1#2#3{\displaystyle\frac{\partial^2{#1}}{\partial{#2}\partial{#3}}}
\def\restric#1#2{\left.#1\right|_{#2}}

\def\vf{{\mathfrak{X}}}
\def\df{{\mit\Omega}}
\def\Lag{{\cal L}}
\def\Leg{{\cal FL}}

\def\d{{\rm d}}

\def\p{{\rm p}}

\def\Zahl{\mathbb{Z}}
\def\Real{\mathbb{R}}
\def\R{\mathbb{R}}

\def\pr{\operatorname{pr}}
\def\Tan{{\rm T}}
\def\Lie{\mathop{\rm L}\nolimits}
\def\inn{\mathop{i}\nolimits}
\def\Cinfty{{\rm C}^\infty}

\def\tabaddress#1{{\small\it\begin{tabular}[t]{c}#1
\\[1.2ex]\end{tabular}}}
\def\qed{\ifvmode\removelastskip\fi
{\unskip\nobreak\hfil\penalty50\hbox{}\nobreak\hfil \hbox{\vrule
height1.2ex width1.2ex}\parfillskip=0pt \finalhyphendemerits=0
\par\smallskip}}

\title{GEOMETRIC HAMILTON-JACOBI THEORY FOR HIGHER--ORDER AUTONOMOUS SYSTEMS}

\author{ 
{\sc  Leonardo Colombo\thanks{{\bf e}-{\it mail}: leo.colombo@icmat.es} }\\
{\sc  Manuel de Le\'on\thanks{{\bf e}-{\it mail}: mdeleon@icmat.es} }\\
\vspace{5mm}
   \tabaddress{Instituto de Ciencias Matem\'{a}ticas (CSIC-UAM-UC3M-UCM). \\
   C/ Nicol\'{a}s Cabrera 15. Campus Cantoblanco UAM. 28049 Madrid. Spain} \\
{\sc  Pedro Daniel Prieto-Mart\'\i nez\thanks{{\bf e}-{\it mail}: peredaniel@ma4.upc.edu} }\\
{\sc Narciso Rom\'an-Roy\thanks{{\bf e}-{\it mail}: nrr@ma4.upc.edu}}  \\
   \tabaddress{Departamento de Matem\'atica Aplicada IV.\\
Universitat Polit\`ecnica de Catalunya-Barcelona Tech.\\
   Edificio C-3, Campus Norte UPC.
   C/ Jordi Girona 1. 08034 Barcelona. Spain}
}

\begin{document}

\maketitle

\thispagestyle{empty}

\begin{abstract}
The geometric framework for the Hamilton--Jacobi theory
is used to study this theory in the ambient of
higher-order mechanical systems, both in the Lagrangian
and Hamiltonian formalisms.
Thus, we state the corresponding Hamilton--Jacobi equations
in these formalisms and apply our results to analyze some
particular physical examples.
\end{abstract}


\smallskip
\noindent
\vbox{\raggedleft AMS s.\,c.\,(2010): 53C15, 53C80, 70H03, 70H05, 70H20, 70H50.}\null

\markright{\rm  L. Colombo, M. de Le\'on, P.D. Prieto-Mart\'{\i}nez, N. Rom\'{a}n-Roy:
 \sl HJ theory  for higher-order}

\clearpage

\tableofcontents


\section{Introduction}
\label{sec:Introduction}

As it is well known, in classical mechanics, Hamilton--Jacobi theory
is a way to integrate a system of ordinary differential equations
(Hamilton equations) that, through an appropriate canonical
transformation \cite{Arn,JS}, is led to equilibrium. The equation to
be satisfied by the generating function of this transformation is a
partial differential equation whose solution allows us to integrate
the original system. In this respect, Hamilton--Jacobi theory
provides important physical examples of the deep connection between
first-order partial differential equations and systems of
first-order ordinary differential equations. It is also very close,
from the classical side, to the Schr\"odinger equation of quantum
mechanics since a complete solution of the Hamilton--Jacobi equation
allows us to reconstruct an approximate solution of the
Schr\"odinger equation (see for instance \cite{EMS04,MMMcq}). For
these reasons, Hamilton--Jacobi theory has been a matter of
continuous interest and has been studied classically also in other
ambients \cite{Gomis1}.

From the viewpoint of geometric mechanics, the intrinsic formulation
of Hamilton--Jacobi equation is also clear \cite{Ab,LM-87,MMM}.
Nevertheless, in the paper
\cite{art:Carinena_Gracia_Marmo_Martinez_Munoz_Roman06} a new
geometric framework for the Hamilton--Jacobi theory was presented
and the Hamilton--Jacobi equation was formulated both in the
Lagrangian and in the Hamiltonian formalisms of autonomous and
non-autonomous mechanics.
A similar generalization of the Hamilton-Jacobi formalism was outlined in \cite{KV-1993}.

Later on, this new geometric framework was used to state the
Hamilton--Jacobi theory in many different situations, such as
non-holonomic mechanical systems
\cite{HJTeam2,leones1,LOS-12,leones2,OFB-11,blo}, geometric
mechanics on Lie algebroids \cite{BMMP-10} and almost-Poisson
manifolds \cite{LMV-13}, singular systems \cite{LMV-12}, control
theory \cite{BLMMM-12,Wang1,Wang2}, classical field theories
\cite{LMM-09,LMMSV-12,DeLeon_Vilarino} and partial differential equations in general
\cite{Vi-11}, the geometric discretization of the Hamilton--Jacobi
equation \cite{OBL-11}, and others \cite{BLM-12,HJteam-K}.

The aim of this paper is to go ahead in this program, applying this
geometric description of the Hamilton--Jacobi theory
in the ambient of higher-order mechanical systems.
These kind of systems appear in many models in theoretical and mathematical physics;
for instance in the mathematical description of relativistic particles with spin, string theories,
gravitation, Podolsky's electromagnetism and others,
in some problems of fluid mechanics and classical physics, and
in numerical models arising from the geometric discretization of first-order dynamical systems
(see \cite{art:Prieto_Roman11,art:Prieto_Roman12_1}
for a long but non-exhaustive list of references).
The geometric descriptions of these systems
use higher-order tangent and jet bundles
(see, for instance, \cite{proc:Cantrijn_Crampin_Sarlet86,
art:Gracia_Pons_Roman91,art:Gracia_Pons_Roman92,art:Krupkova00,book:DeLeon_Rodrigues85,
art:Prieto_Roman11,art:Prieto_Roman12_1,art:Prieto_Roman13}).
Up to our knowledge, the Hamilton--Jacobi equation for these kind of systems
has deserved little attention
(although an approach to the problem for
higher-order field theories can be found in \cite{Vi-12}).
This paper wants to fill this gap.

Thus, after establishing some basic concepts and notation
about higher-order tangent and jet bundles in Section \ref{sec:ReviewHOTangentBundles},
we state in Section \ref{sec:LagFormalism} the generalized
and the standard Hamilton--Jacobi problems for the Lagrangian
formalism of higher-order systems, including a discussion about
complete solutions to both problems and detailed coordinate expressions
of the Lagrangian Hamilton--Jacobi equations so obtained.
Section \ref{HamHamProb} is devoted to present
the Hamilton--Jacobi problems and their complete solutions
in the Hamiltonian formalism, as a straightforward application of
the corresponding case for first-order mechanics stated in
\cite{art:Carinena_Gracia_Marmo_Martinez_Munoz_Roman06}
and some results of the Lagrangian case.
Finally, in Section \ref{sec:Examples}, our model is applied
to two real physical examples:
the system describing the dynamics of the end of a javelin
and a (homogeneous) deformed elastic cylindrical beam with fixed ends.

All the manifolds are real, second countable and $\Cinfty$.
The maps and the structures are assumed to be $\Cinfty$.
Sum over repeated indices is understood.

\section{Higher-order tangent bundles and canonical structures}
\label{sec:ReviewHOTangentBundles}


(See \cite{book:DeLeon_Rodrigues85}, \cite{book:Saunders89} for details).

Let $Q$ be a $n$-dimensional manifold, and $k \in \Zahl^+$. The
{\sl $k$th order tangent bundle of $Q$}
is the $(k+1)n$-dimensional manifold $\Tan^{k}Q$ made of the $k$-jets
of the bundle $\pi \colon \R \times Q \to \R$
with fixed source point $t = 0 \in \R$; that is, $\Tan^kQ = J_0^k\pi$.
It is a $1$-codimensional submanifold of $J^k\pi$.

We have the following natural projections (for $r \leqslant k$):
$$
\begin{array}{rcl}
\rho^k_r \colon \Tan^kQ & \longrightarrow & \Tan^rQ \\
j^k_0\phi & \longmapsto & j^r_0\phi
\end{array} \quad ; \quad
\begin{array}{rcl}
\beta^k \colon \Tan^kQ & \longrightarrow & Q \\
j^k_0\phi & \longmapsto & \phi(0)
\end{array}
$$
where $j^k_0\phi$ denotes a point in $\Tan^kQ$; that is,
the equivalence class of a curve $\phi \colon I \subset \R \to Q$
by the $k$-jet equivalence relation. Notice that $\rho^k_0 = \beta^k$,
where $\Tan^0Q$ is canonically identified with $Q$, and $\rho^k_k = {\rm Id}_{\Tan^kQ}$.
Observe also that $\rho^l_s \circ \rho^r_l = \rho^r_s$, for
$0 \leqslant s \leqslant l \leqslant r \leqslant k$.

 The natural projections $\rho^r_s
\colon \Tan^rQ \to \Tan^sQ$ are smooth surjective submersions and,
furthermore, the triple $(\Tan^rQ,\rho^r_s,\Tan^sQ)$ is a smooth
fiber bundle with fiber $\R^{(r-s)n}$ (see \cite{book:Saunders89}).
In particular, $(\Tan^kQ,\rho^k_r,\Tan^rQ)$ is a smooth fiber bundle
with fiber $\R^{(k-r)n}$, for $0\leqslant r \leqslant k$; that is,
$\Tan^kQ$ is canonically endowed with $k+1$ different affine bundle
structures given by the projections
$\rho^k_0,\rho^k_1,\ldots,\rho^k_k$. In the sequel, we refer to this
fiber bundle structure as {\sl the $\rho^k_r$-bundle structure of
$\Tan^kQ$}.

If $\left(U,\varphi\right)$ is a local chart in $Q$,
with $\varphi = (\varphi^A)$, $1\leqslant A \leqslant n$,
and $\phi \colon \R \to Q$ is a curve in $Q$ such that $\phi(0) \in U$;
by writing $\phi^A = \varphi^A \circ \phi$, the $k$-jet $j^k_0\phi$
is given in $\left(\beta^k\right)^{-1}(U) = \Tan^kU$ by
$\left(q^A,q^A_{1},\ldots,q^A_{k}\right)$, where
$$
q^A = \phi^A(0) \quad ; \quad
q_{i}^A = \restric{\frac{d^i\phi^A}{dt^i}}{t=0} \, .
$$
with $1\leqslant i \leqslant k$. Usually we write $q_{0}^A$ instead of $q^A$,
and so we have local coordinates $\left(q_{0}^A,q_{1}^A,\ldots,q_{k}^A\right) = \left(q_i^A\right)$,
with $0 \leqslant i \leqslant k$, in the open set $\left(\beta^k\right)^{-1}(U) \subseteq \Tan^kQ$.
Using these coordinates, the local expression of the canonical projections are
$$
\rho^k_r\left(q_0^A,q_1^A,\ldots,q_k^A\right) = \left(q_0^A,q_1^A,\ldots,q_r^A\right) \quad ; \quad
\beta^k\left(q_0^A,q_1^A,\ldots,q_k^A\right) = \left(q_0^A\right) \, .
$$
Hence, local coordinates in the open set $\left(\beta^k\right)^{-1}(U) \subseteq \Tan^kQ$
adapted to the $\rho^k_r$-bundle structure are $\left(q_0^A,\ldots,q_r^A,q_{r+1}^A,\ldots,q_{k}^A\right)$,
and a section $s \in \Gamma(\rho^k_r)$ is locally given in this open set by
$s(q_0^A,\ldots,q_r^A) = \left( q_0^A,\ldots,q_r^A,s_{r+1}^A,\ldots,s_k^A \right)$,
where $s_j^A$ (with $r+1 \leqslant j \leqslant k$) are local functions.

If $\phi \colon \R \to Q$ is a curve in $Q$,
the {\rm canonical lifting} of $\phi$ to $\Tan^kQ$
is the curve $j^k\phi\colon \Real\to\Tan^kQ$ defined as
the $k$-jet lifting of $\phi$ restricted to the submanifold
$\Tan^{k}Q \hookrightarrow J^{k}\pi$


Now, consider the bundle $\pi \colon \R \times Q \to \R$ and the natural projection
 $\pi^r_s \colon J^{r}\pi \to J^{s}\pi$, for $0 \leqslant s \leqslant r$.
Let $\vf(\pi^r_s)$ be the module of vector fields along the projection $\pi^r_s$.

\begin{definition}\label{def:VFHolonomicLift}
If $X \in \vf(\R)$, $\phi \in \Gamma(\pi)$ and $t_o \in \R$,
the {\rm $k$th holonomic lift} of $X$ by $\phi$ is defined as
the vector field $j^{k}X \in \vf(\pi^{k+1}_{k})$ whose associated derivation satisfies
$$
(d_{j^kX}f)(j^{k+1}\phi(t_o)) = d_X(f \circ j^{k}\phi)({t_o}) \, ,
$$
for every $f \in \Cinfty(J^k\pi)$; where $d_{j^kX}$ and $d_{X}$ are the derivations
associated to $j^kX$ and $X$, respectively.
\end{definition}

In local coordinates, if $X \in \vf(\R)$ is given by
$\ds X = X_o \derpar{}{t}$,
then the $k$th holonomic lift of $X$ is
$$
j^kX = X_o\left( \derpar{}{t} + \sum_{i=0}^{k} q_{i+1}^A\derpar{}{q_i^A}\right) \, .
$$

The \textsl{(non-autonomous) total time derivative} is the derivation associated to the
$k$th holonomic lift of the coordinate vector field $\ds\frac{\partial}{\partial t} \in \vf(\R)$,
which is denoted by $\ds\frac{d}{dt} \in \vf(\pi^{k+1}_{k})$, and
whose local expression is
$$
\frac{d}{dt} = \derpar{}{t} + \sum_{i=0}^{k}q_{i+1}^A\derpar{}{q^A_{i}} \, .
$$
Using the identification $J^k\pi \cong \R \times \Tan^kQ$ and denoting by
$\pi_2 \colon \R \times Q \to Q$ the natural projection and
all the induced projections in higher-order jet bundles,
we have the following diagram
$$
\xymatrix{
\R \times \Tan^{k+1}Q \ar[dd]_{\pi^{k+1}_{k}} \ar[rr]^-{\pi_2} & \ & \Tan^{k+1}Q \ar[dd]^{\rho^{k+1}_{k}} \\
\ & \ & \ \\
\R \times \Tan^{k}Q \ar[rr]^-{\pi_2} & \ & \Tan^{k}Q
}
$$
Hence, the total time derivative induces a vector field $T \in \vf(\rho^{k+1}_{k})$,
which satisfies
$$
T_{\pr_2(j^{k+1}\phi)} = \Tan\pi_2 \circ \restric{\frac{d}{dt}}{j^{k+1}\phi}  \, ,
$$
and the derivation associated to the vector field $T$, denoted by $d_T$,
is called {\sl the (autonomous) total time derivative}, or {\sl Tulczyjew's derivation}.

The total time derivative in higher-order tangent bundles can be
introduced in an equivalent way without using explicitly the jet
bundle structure (see \cite{Tu-75} for details).


\begin{definition}
A curve $\psi \colon \R \to \Tan^{k}Q$ is {\rm holonomic of type $r$}, $1 \leqslant r \leqslant k$,
if $j^{k-r+1}\phi = \rho^{k}_{k-r+1} \circ \psi$,
where $\phi = \beta^{k} \circ \psi \colon \R \to Q$;
that is, the curve $\psi$ is the lifting of a curve in $Q$ up to $\Tan^{k-r+1}Q$.
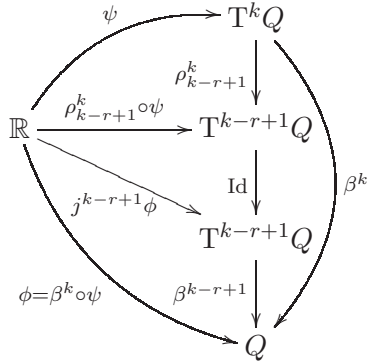
$$
\xymatrix{
\ & \ & \Tan^kQ \ar[d]_{\rho^k_{k-r+1}} \ar@/^2.5pc/[ddd]^{\beta^k} \\
\R \ar@/^1.5pc/[urr]^{\psi} \ar@/_1.5pc/[ddrr]_{\phi = \beta^k\circ\psi}
\ar[rr]^-{\rho^k_{k-r+1}\circ\psi} \ar[drr]_{j^{k-r+1}\phi}
& \ & \Tan^{k-r+1}Q \ar[d]_{{\rm Id}} \\
\ & \ & \Tan^{k-r+1}Q \ar[d]_{\beta^{k-r+1}} \\
\ & \ & Q
}
$$
In particular, a curve $\psi$ is {\rm holonomic of type 1}
if $j^{k}\phi = \psi$, with $\phi = \beta^{k} \circ \psi$.
Throughout this paper, holonomic curves of type $1$ are simply called {\rm holonomic}.
\end{definition}

\begin{definition}
A vector field $X \in \vf(\Tan^kQ)$ is a {\rm semispray of type $r$}, $1 \leqslant r \leqslant k$,
if every integral curve $\psi$ of $X$ is holonomic of type $r$.
\end{definition}

The local expression of a semispray of type $r$ is
$$
X = q_1^A\derpar{}{q_0^A} + q_2^A\derpar{}{q_1^A} + \ldots + q_{k-r+1}^A\derpar{}{q_{k-r}^A} +
F_{k-r+1}^A\derpar{}{q_{k-r+1}^A} + \ldots + F_k^A\derpar{}{q_k^A} \, .
$$

It is clear that every holonomic curve of type $r$ is also holonomic of
type $s$, for $s \geqslant r$. The same remark is true for
semisprays.

Observe that semisprays of type $1$ in
$\Tan^kQ$ are the analogue to the holonomic  (or SODE) vector fields in first-order mechanics.
Their local expressions are
$$
X = q_1^A\derpar{}{q_0^A} + q_2^A\derpar{}{q_1^A} + \ldots + q_{k}^A\derpar{}{q_{k-1}^A} + F^A\derpar{}{q_k^A} \, .
$$

If $X\in\vf(\Tan^kQ)$ is a semispray of type $r$, a
curve $\phi \colon \R \to Q$ is said to be a {\sl path} or {\sl
solution} of $X$ if $j^k\phi$ is an integral curve of $X$; that is,
$\widetilde{j^k\phi} = X \circ j^k\phi$, where $\widetilde{j^k\phi}$
denotes the canonical lifting of $j^k\phi$ from $\Tan^kQ$ to
$\Tan(\Tan^kQ)$. Then, in coordinates, $\phi$ verifies the following
system of differential equations of order $k+1$:
$$
\frac{d^{k-r+2}\phi^A}{d t^{k-r+2}} = X_{k-r+1}^A\left(\phi,\frac{d\phi}{d t},\ldots,\frac{d^k\phi}{d t^k}\right)
, \ \ldots \ ,
\frac{d^{k+1}\phi^A}{d t^{k+1}} = X_k^A\left(\phi,\frac{d\phi}{d t},\ldots,\frac{d^k\phi}{d t^k}\right) \ .
$$

\section{The Hamilton-Jacobi problem in the Lagrangian formalism}
\label{sec:LagFormalism}

Let $Q$ be the configuration space of an autonomous dynamical system
of order $k$ with $n$ degrees of freedom (that is, $Q$ is a
$n$-dimensional smooth manifold), and let $\Lag \in
\Cinfty(\Tan^{k}Q)$ be the Lagrangian function for this system. From
the Lagrangian function $\Lag$ we construct the Poincar\'{e}-Cartan
forms $\theta_{\Lag}\in\df^{1}(\Tan^{2k-1}Q)$ and
$\omega_\Lag = -\d\theta_{\Lag}\in \df^{2}(\Tan^{2k-1}Q)$,
as well as the Lagrangian energy $E_\Lag \in \Cinfty(\Tan^{2k-1}Q)$.
The coordinate
expressions of these elements are
\begin{align*}
&\theta_{\Lag} = \sum_{r=1}^k \sum_{i=0}^{k-r}(-1)^i d_T^i\left(\derpar{\Lag}{q_{r+i}^A}\right) \d q_{r-1}^A \, , \\
&\omega_\Lag =
\sum_{r=1}^k \sum_{i=0}^{k-r}(-1)^{i+1} d_T^i\,\d\left(\derpar{\Lag}{q_{r+i}^A}\right) \wedge \d q_{r-1}^A \, , \\
&E_\Lag = \sum_{r=1}^{k}
q_{r}^A \sum_{i=0}^{k-r} (-1)^i d_T^i\left( \derpar{\Lag}{q_{r+i}^A}
\right)- \Lag \, .
\end{align*}

We assume that the Lagrangian function is regular, and then
$\omega_\Lag$ is a symplectic form. Then the dynamical equation for
this Lagrangian system, which is
\begin{equation}\label{eqn:LagDynEq}
\inn(X_\Lag) \, \omega_\Lag = \d E_\Lag \, ,
\end{equation}
has a unique solution $X_\Lag \in \vf(\Tan^{2k-1}Q)$ which
is a semispray of type $1$ in $\Tan^{2k-1}Q$.
(For a detailed description of the Lagrangian and Hamiltonian formalisms
of higher-order dynamical systems see, for instance,
\cite{book:DeLeon_Rodrigues85,art:Prieto_Roman11,art:Prieto_Roman12_1,art:Prieto_Roman13}).

\subsection{The generalized Lagrangian Hamilton-Jacobi problem}
\label{sec:GenLagHJProb}

Following \cite{Ab} and
\cite{art:Carinena_Gracia_Marmo_Martinez_Munoz_Roman06}, we first
state a general version of the Hamilton-Jacobi problem in the
Lagrangian setting, which is the so-called \textsl{generalized
Lagrangian Hamilton-Jacobi problem}. For first-order systems, this
problem consists in finding vector fields $X \in \vf(Q)$ such that
the liftings of any integral curve of $X$ to $\Tan Q$ by $X$ itself
is an integral curve of the Lagrangian vector field $X_\Lag$. In
higher-order systems we can state an analogous problem. Thus, we
have the following definition:

\begin{definition}\label{def:GenLagHJDef}
The \textnormal{generalized $k$th-order Lagrangian Hamilton-Jacobi problem} consists in finding a section
$s \in \Gamma(\rho^{2k-1}_{k-1})$ and a vector field $X \in \vf(\Tan^{k-1}Q)$ such that,
if $\gamma \colon \R \to \Tan^{k-1}Q$ is an integral curve of $X$,
then $s \circ \gamma \colon \R \to \Tan^{2k-1}Q$ is an integral curve of $X_\Lag$;
that is,
\begin{equation}\label{eqn:GenLagHJDef}
X \circ \gamma = \dot \gamma \Longrightarrow X_\Lag \circ (s \circ \gamma) = \dot{\overline{s \circ \gamma}} \, .
\end{equation}
\end{definition}

\noindent\textbf{Remark:}
Observe that, since $X_\Lag$ is a semispray of type $1$, then every integral curve of $X_\Lag$
is the $(2k-1)$-jet lifting of a curve in $Q$. In particular, this holds for
the curve $s \circ \gamma$, that is, there exists a curve $\phi \colon \R \to Q$
such that
$$
j^{2k-1}\phi = s \circ \gamma \, .
$$
Then, composing both sides of the equality with $\rho^{2k-1}_{k-1}$ and bearing
in mind that $s \in \Gamma(\rho^{2k-1}_{k-1})$, we obtain
$$
\gamma = j^{k-1}\phi \, ,
$$
that is, the curve $\gamma$ is the $(k-1)$-jet lifting of a curve in $Q$. This enables us to restate the problem
as follows: \textit{The \textnormal{generalized $k$th-order Lagrangian Hamilton-Jacobi problem} consists
in finding a vector field $X_o \in \vf(Q)$ such that, if $\phi \colon \R \to Q$ is an integral curve of $X_o$,
then $j^{2k-1}\phi \colon \R \to \Tan^{2k-1}Q$ is an integral curve of $X_\Lag$; that is,}
$$
X_o \circ \phi = \dot{\phi} \Longrightarrow X_\Lag \circ (j^{2k-1}\phi)) = \dot{\overline{j^{2k-1}\phi}} \, .
$$
Nevertheless, we will stick to the previous statement (Definition \ref{def:GenLagHJDef}) in order to give several
different characterizations of the problem.

It is clear from Definition \ref{def:GenLagHJDef} that the vector field $X \in \vf(\Tan^{k-1}Q)$ cannot be
chosen independently from the section $s \in \Gamma(\rho^{2k-1}_{k-1})$. In fact:

\begin{prop}\label{prop:GenLagHJRelatedVF}
The pair $(s,X) \in \Gamma(\rho^{2k-1}_{k-1}) \times \vf(\Tan^{k-1}Q)$ satisfies the condition
\eqref{eqn:GenLagHJDef} if, and only if, $X$ and $X_\Lag$ are $s$-related; that is,
$X_\Lag \circ s = \Tan s \circ X$.
\end{prop}
\begin{proof}
If $(s,X)$ satisfies the condition \eqref{eqn:GenLagHJDef}, then for every integral curve $\gamma$ of
$X,$ we have
$$
X_\Lag \circ ( s \circ \gamma) = \dot{\overline{s \circ \gamma}}
= \Tan s \circ \dot\gamma = \Tan s \circ X \circ \gamma \, ,
$$
but, as $X$ has integral curves through every point $\bar{y} \in \Tan^{k-1}Q$,
this is equivalent to $X_\Lag \circ s = \Tan s \circ X$.

Conversely, if $X_{\Lag}$ and $X$ are
$s$-related and $\gamma\colon\R\rightarrow\Tan^{k-1}Q$ is an
integral curve of $X$, we have
$$
X_\Lag \circ s \circ \gamma = \Tan s\circ X \circ \gamma = \Tan s\circ\dot{\gamma} = \dot{\overline{s \circ \gamma}} \, .
$$
\end{proof}

Hence, the vector field $X \in \vf(\Tan^{k-1}Q)$ is related with the Lagrangian
vector field $X_\Lag$ and with the section $s \in \Gamma(\rho^{2k-1}_{k-1})$.
As a consequence of Proposition \ref{prop:GenLagHJRelatedVF}, composing
both sides of the equality $X_\Lag \circ s = \Tan s \circ X$ with $\Tan\rho^{2k-1}_{k-1}$
and bearing in mind that $\rho^{2k-1}_{k-1} \circ s = \textnormal{Id}_{\Tan^{k-1}Q}$, we have:

\begin{corol}\label{corol:GenLagHJRelatedVF}
If $(s,X)$ satisfies condition \eqref{eqn:GenLagHJDef}, then
$X = \Tan\rho^{2k-1}_{k-1} \circ X_\Lag \circ s$.
\end{corol}

Thus, the vector field $X$ is completely determined by the section
$s \in \Gamma(\rho^{2k-1}_{k-1})$, and it is called the
\textsl{vector field associated to $s$}. The
following diagram illustrates the situation
$$
\xymatrix{
\Tan(\Tan^{k-1}Q) \ar@/_1.5pc/[rr]_{\Tan s}  & \ & \ar[ll]_{\Tan\rho^{2k-1}_{k-1}} \Tan(\Tan^{2k-1}Q) \\
\ & \ & \ \\
\Tan^{k-1}Q \ar[uu]^{X} \ar@/_1.5pc/[rr]_{s} & \ & \ar[ll]_{\rho^{2k-1}_{k-1}} \Tan^{2k-1}Q \ar[uu]_{X_\Lag}
}
$$

Since the vector field $X$ is completely determined by the section $s$, the search of a pair
$(s,X) \in \Gamma(\rho^{2k-1}_{k-1}) \times \vf(\Tan^{k-1}Q)$ satisfying condition
\eqref{eqn:GenLagHJDef} is equivalent to the search of a section $s \in \Gamma(\rho^{2k-1}_{k-1})$
such that the pair $(s,\Tan\rho^{2k-1}_{k-1} \circ X_\Lag \circ s)$ satisfies the same condition.
Thus, we can give the following definition:

\begin{definition}
A \textnormal{solution to the generalized $k$th-order Lagrangian Hamilton-Jacobi problem}
for $X_\Lag$ is a section $s \in \Gamma(\rho^{2k-1}_{k-1})$ such that, if
$\gamma \colon \R \to \Tan^{k-1}Q$ is an integral curve of
$\Tan\rho^{2k-1}_{k-1} \circ X_\Lag \circ s$, then
$s \circ \gamma \colon \R \to \Tan^{2k-1}Q$ is an integral curve of
$X_\Lag$, that is,
$$
\Tan\rho^{2k-1}_{k-1} \circ X_\Lag \circ s \circ \gamma = \dot{\gamma}
\Longrightarrow X_\Lag \circ (s \circ \gamma) = \dot{\overline{s \circ \gamma}} \, .
$$
\end{definition}

Finally, we have the following result, which gives some equivalent conditions
for a section to be a solution to the generalized $k$th-order Lagrangian Hamilton-Jacobi problem:

\begin{prop}\label{prop:GenLagHJEquiv}
The following assertions on a section $s \in \Gamma(\rho^{2k-1}_{k-1})$
are equivalent.
\begin{enumerate}
\item The section $s$ is a solution to the generalized $k$th-order Lagrangian Hamilton-Jacobi problem.
\item The submanifold ${\rm Im}(s) \hookrightarrow \Tan^{2k-1}Q$
is invariant by the Euler-Lagrange vector field
$X_\Lag$ (that is, $X_\Lag$ is tangent to the submanifold
$s(\Tan^{k-1}Q) \hookrightarrow \Tan^{2k-1}Q$).
\item The section $s$ satisfies the equation
$$
\inn(X)(s^*\omega_\Lag) = \d(s^*E_\Lag) \, .
$$
where $X = \Tan\rho^{2k-1}_{k-1} \circ X_\Lag \circ s$ is the vector field associated to $s$.
\end{enumerate}
\end{prop}
\begin{proof}

\noindent ($1 \Longleftrightarrow 2$) Let $s$ be a solution to the
generalized $k$th-order Lagrangian Hamilton-Jacobi problem. Then by
Proposition \ref{prop:GenLagHJRelatedVF} the Lagrangian vector field
$X_{\Lag} \in \vf(\Tan^{2k-1}Q)$ is $s$-related to the
vector field $X = \Tan\rho^{2k-1}_{k-1} \circ X_\Lag \circ s \in \vf(\Tan^{k-1}Q)$
associated to $s$, and thus for every $\bar{y} \in \Tan^{k-1}Q$ we have
$$
X_{\Lag}(s(\bar{y})) = (X_{\Lag}\circ s)(\bar{y}) = (\Tan s \circ X)(\bar{y}) = \Tan s(X(\bar{y})) \, .
$$
Hence, $X_{\Lag}(s(\bar{y}))=\Tan s(X(\bar{y}))$ and therefore
$X_{\Lag}$ is tangent to the submanifold $\textnormal{Im}(s)\hookrightarrow
\Tan^{2k-1}Q.$

Conversely, if the submanifold $\textnormal{Im}(s)$ is
invariant under the flow of $X_{\Lag}$, then $X_{\Lag}(s(\bar{y}))
\in \Tan_{s(\bar{y})}\textnormal{Im}(s)$, for every $\bar{y} \in \Tan^{k-1}Q$;
that is, there exists an element $u_{\bar{y}} \in \Tan_{\bar{y}}\Tan^{k-1}Q$
such that $X_{\Lag}(s(\bar{y})) = \Tan_{\bar{y}}s(u_{\bar{y}})$. If we
define $X \in \vf(\Tan^{k-1}Q)$ as the vector field that satisfies
$\Tan_{\bar{y}}s(X_{\bar{y}}) = X_{\Lag}(s(\bar{y}))$, then $X$ is a
vector field in $\Tan^{k-1}Q$, since $X = \Tan\rho^{2k-1}_{k-1} \circ
X_\Lag \circ s$, and it is $s$-related with $X_{\Lag}$.
Therefore, by Proposition \ref{prop:GenLagHJRelatedVF}, $s$ is a solution to the generalized
$k$th-order Lagrangian Hamilton-Jacobi problem.

\noindent ($1 \Longleftrightarrow 3$) Let $s$ be a solution to the
generalized $k$th-order Lagrangian Hamilton-Jacobi problem. Taking the pull-back of
the Lagrangian dynamical equation \eqref{eqn:LagDynEq} by the section $s$ we have
$$
s^*\inn(X_{\Lag})\omega_{\Lag} = s^*\d E_{\Lag} = \d(s^*E_{\Lag}) \, ,
$$
but since $X$ and $X_{\Lag}$ are $s$-related by Proposition \ref{prop:GenLagHJRelatedVF},
we have that $s^*\inn(X_{\Lag})\omega_{\Lag} = \inn(X)s^*\omega_{\Lag}$, and hence
we obtain
$$
\inn(X)s^*\omega_{\Lag} = \d(s^*E_{\Lag}) \, .
$$

Conversely, consider the following vector field along
the section $s\in\Gamma(\rho_{k-1}^{2k-1})$
$$
D_\Lag = X_{\Lag} \circ s - \Tan s \circ X \colon \Tan^{k-1}Q \to \Tan(\Tan^{2k-1}Q) \, .
$$
We want to prove that $D_\Lag = 0$ or equivalently, as $\omega_{\Lag}$ is non-degenerate, that
$(\omega_{\Lag})_{s(\bar{y})}(D_\Lag(\bar{y}),Z_{s(\bar{y})}) = 0$,
for every tangent vector $Z_{s(\bar{y})} \in \Tan_{s(\bar{y})}\Tan^{2k-1}Q$.
Taking the pull-back of the Lagrangian dynamical equation, and using the hypothesis, we have
$$
s^{*}(\inn(X_{\Lag})\omega_{\Lag}) = s^{*}\d E_{\Lag} = \d(s^{*}E_{\Lag})=\inn(X)(s^{*}\omega_{\Lag}) \, ,
$$
and then the form
$s^{*}(\inn(X_{\Lag})\omega_{\Lag})-\inn(X)(s^{*}\omega_{\Lag}) \in \df^1(\Tan^{k-1}Q)$ vanishes.
Therefore, for every $\bar{y}\in \Tan^{k-1}Q$ and $u_{\bar{y}}\in \Tan_{\bar{y}}\Tan^{k-1}Q$, we have
\begin{align*}
0 &=
(s^*\inn(X_{\Lag})\omega_{\Lag} - \inn(X)s^*\omega_{\Lag})_{\bar{y}}(u_{\bar{y}}) \\
&= (\omega_{\Lag})_{s(\bar{y})}(X_{\Lag}(s(\bar{y})),\Tan_{\bar{y}}s(u_{\bar{y}}))
- (\omega_{\Lag})_{s(\bar{y})}(\Tan_{\bar{y}}s(X_{\bar{y}}),\Tan_{\bar{y}}s(u_{\bar{y}})) \\
&= (\omega_{\Lag})_{s(\bar{y})}(X_{\Lag}(s(\bar{y})) - \Tan_{\bar{y}}s(X_{\bar{y}}),\Tan_{\bar{y}}s(u_{\bar{y}})) \\
&= (\omega_{\Lag})_{s(\bar{y})}(D_\Lag(\bar{y}),\Tan_{\bar{y}}s(u_{\bar{y}})) \, .
\end{align*}
Therefore,
$(\omega_{\Lag})_{s(\bar{y})}(D_\Lag(\bar{y}),A_{s(\bar{y})}) = 0$, for
every $A_{s(\bar{y})} \in \Tan_{s(\bar{y})}{\rm Im}(s)$. Now recall that
every section defines a canonical splitting of the tangent space
of $\Tan^{2k-1}Q$ at every point given by
$$
\Tan_{s(\bar{y})}\Tan^{2k-1}Q = \Tan_{s(\bar{y})}{\rm Im}(s)
\oplus V_{s(\bar{y})}(\rho^{2k-1}_{k-1}) \, .
$$
Thus, we only need to prove that
$(\omega_{\Lag})_{s(\bar{y})}(D_\Lag(\bar{y}),B_{s(\bar{y})}) = 0$, for
every vertical tangent vector $B_{s(\bar{y})} \in
V_{s(\bar{y})}(\rho^{2k-1}_{k-1})$. Equivalently, as $\omega_{\Lag}$
is annihilated by the contraction of two $\rho^{2k-1}_{k-1}$-vertical
vectors, it suffices to prove that $D_\Lag$ is vertical with respect to
that submersion. Indeed,
\begin{align*}
\Tan\rho^{2k-1}_{k-1} \circ D_\Lag &=
\Tan\rho^{2k-1}_{k-1} \circ (X_{\Lag} \circ s - \Tan s \circ X) \\
&= \Tan\rho^{2k-1}_{k-1} \circ X_{\Lag} \circ s - \Tan\rho^{2k-1}_{k-1} \circ \Tan s \circ X \\
&= \Tan\rho^{2k-1}_{k-1} \circ X_{\Lag} \circ s - \Tan(\rho^{2k-1}_{k-1}\circ s) \circ X \\
&= \Tan\rho^{2k-1}_{k-1} \circ X_{\Lag} \circ s - X = 0
\end{align*}
Therefore
$(\omega_{\Lag})_{s(\bar{y})}(D_\Lag(\bar{y}),Z_{s(\bar{y})}) = 0$, for
every $Z_{s(\bar{y})} \in \Tan_{s(\bar{y})}\Tan^{2k-1}Q$, and as
$\omega_{\Lag}$ is non-degenerate, we have that $X_{\Lag}$ and $X$
are $s$-related, and by Proposition \ref{prop:GenLagHJRelatedVF} $s$
is a solution to the generalized $k$th-order Lagrangian Hamilton-Jacobi problem.
\end{proof}

Observe that if $s \in \Gamma(\rho^{2k-1}_{k-1})$ is a solution to the generalized
$k$th-order Lagrangian Hamilton-Jacobi problem then, taking into account Corollary \ref{corol:GenLagHJRelatedVF},
we can conclude that the integral curves of the Lagrangian vector field $X_\Lag$
contained in $\textnormal{Im}(s)$ project to $\Tan^{k-1}Q$ by $\rho^{2k-1}_{k-1}$ to integral curves of
$\Tan^{2k-1}_{k-1}\circ X_\Lag \circ s$. The converse, however, is not true
unless we assume further assumptions.

\noindent{\bf Remark:}
Notice that, except for the third item in Proposition \ref{prop:GenLagHJEquiv},
all the results stated in this Section hold for every vector field $Z \in \vf(\Tan^{2k-1}Q)$,
not only for the Lagrangian vector field $X_\Lag$. Indeed, the
assumption for $X_\Lag$ being the Lagrangian vector field solution to
the equation \eqref{eqn:LagDynEq} is only needed to prove that
the section $s \in \Gamma(\rho^{2k-1}_{k-1})$ and its associated
vector field $X \in \vf(\Tan^{k-1}Q)$ satisfy some kind of dynamical equation in $\Tan^{k-1}Q$.

\paragraph{Coordinate expression:}
Let $(q_0^A)$ be local coordinates in $Q$, and $(q_0^A,\ldots,q_{2k-1}^A)$ the induced
natural coordinates in $\Tan^{2k-1}Q$. Then,
local coordinates in $\Tan^{2k-1}Q$ adapted to the $\rho^{2k-1}_{k-1}$-bundle
structure are $(q_0^A,\ldots,q_{k-1}^A;q_{k}^A,\ldots,q_{2k-1}^A) \equiv (q_i^A;q_j^A)$,
where $i = 0,\ldots,k-1$ and $j = k,\ldots,2k-1$. Until the end of this section,
the indices $i,j$ will take the values above.
Hence, a section $s \in \Gamma(\rho^{2k-1}_{k-1})$ is given locally by
$s(q_0^A,\ldots,q_{k-1}^A) = (q_0^A,\ldots,q_{k-1}^A,s_{k}^A,\ldots,s_{2k-1}^A) \equiv (q_i^A,s_j^A)$,
where $s_j^A$ are local smooth functions in $\Tan^{k-1}Q$.

Let us check what is the local condition for a section $s \in \Gamma(\rho^{2k-1}_{k-1})$
to be a solution to the generalized $k$th-order Lagrangian Hamilton-Jacobi problem.
By Proposition \ref{prop:GenLagHJEquiv},
this is equivalent to require the Lagrangian vector field $X_\Lag \in \vf(\Tan^{2k-1}Q)$
to be tangent to the submanifold $\textnormal{Im}(s) \hookrightarrow \Tan^{2k-1}Q$.
As $\textnormal{Im}(s)$ is locally defined by the constraints
$q_j^A - s_j^A = 0$, we must require
$\Lie(X_\Lag)(q_j^A-s_j^A) \equiv X_\Lag(q_j^A-s_j^A) = 0$
(on $\textnormal{Im}(s)$), for $k \leqslant j \leqslant 2k-1$,
$1 \leqslant A \leqslant n$.
From \cite{book:DeLeon_Rodrigues85,art:Prieto_Roman11}
we know that the Lagrangian vector field $X_\Lag$ has the following
local expression
$$
X_\Lag = q_1^A\derpar{}{q_0^A} + q_2^A\derpar{}{q_1^A} +
\ldots + q_{2k-1}^A\derpar{}{q_{2k-2}^A} + F^A \derpar{}{q_{2k-1}^A} \, ,
$$
where $F^A$ are the functions solution to the following system of $n$ equations
\begin{equation}\label{eqn:LagVFLastComponents}
(-1)^k(F^B - d_T(q_{2k-1}^B))\derpars{\Lag}{q_k^B}{q_k^A}
+ \sum_{l=0}^k(-1)^l d_T^l\left(\derpar{\Lag}{q_l^A}\right) = 0 \, .
\end{equation}
Hence, the condition $\restric{X_\Lag(q_j^A-s_j^A)}{\textnormal{Im}(s)} = 0$ gives the following equations
\begin{equation}\label{eqn:GenLagHJLocal}
s_{j+1}^A - \sum_{i=0}^{k-2}q_{i+1}^B\derpar{s_j^A}{q_i^B} - s_{k}^B\derpar{s_j^A}{q_{k-1}^B} = 0 \quad ; \quad
\displaystyle \restric{F^A}{{\rm Im}(s)} - \sum_{i=0}^{k-2} q_{i+1}^B\derpar{s_{2k-1}^A}{q_i^B} - s_{k}^B\derpar{s_j^A}{q_{k-1}^B} = 0 \, ,
\end{equation}
This is a system of $kn$ partial differential equations
with $kn$ unknown functions $s_j^A$.
Thus, a section $s \in \Gamma(\rho^{2k-1}_{k-1})$ solution
to the generalized $k$th-order Lagrangian Hamilton-Jacobi problem must satisfy
the local equations \eqref{eqn:GenLagHJLocal}.

\subsection{The Lagrangian Hamilton-Jacobi problem}
\label{sec:LagHJProb}

In general, to solve the generalized $k$th-order Lagrangian Hamilton-Jacobi problem can be a difficult task,
since it amounts to find $kn$-dimensional submanifolds of $\Tan^{2k-1}Q$
invariant by the Lagrangian
vector field $X_\Lag$, or, equivalently, solutions to a large
system of partial differential equations with many
unknown functions. Therefore, in order to simplify the problem,
 it is convenient to impose some additional conditions to the section
$s \in \Gamma(\rho^{2k-1}_{k-1})$, thus considering a less general problem.

\begin{definition}\label{def:LagHJProb}
The \textnormal{$k$th-order Lagrangian Hamilton-Jacobi problem} consists in finding
sections $s \in \Gamma(\rho^{2k-1}_{k-1})$ solution to the generalized
$k$th-order Lagrangian Hamilton-Jacobi problem satisfying that
$s^*\omega_\Lag = 0$. Such a section is called a
\textnormal{solution to the $k$th-order Lagrangian Hamilton-Jacobi problem}.
\end{definition}

With the new assumption in Definition \ref{def:LagHJProb},
a straightforward consequence of Proposition \ref{prop:GenLagHJEquiv}  is:

\begin{prop}\label{prop:LagHJEquiv}
Let $s \in \Gamma(\rho^{2k-1}_{k-1})$ such that $s^*\omega_\Lag = 0$. The following assertions
are equivalent:
\begin{enumerate}
\item The section $s$ is a solution to the $k$th-order Lagrangian Hamilton-Jacobi problem.
\item $\d(s^*E_\Lag) = 0$.
\item $\textnormal{Im}(s)$ is a Lagrangian submanifold of $\Tan^{2k-1}Q$ invariant by $X_\Lag$.
\item The integral curves of $X_\Lag$ with initial conditions in $\textnormal{Im}(s)$ project
onto the integral curves of $X = \Tan\rho^{2k-1}_{k-1} \circ X_\Lag \circ s$.
\end{enumerate}
\end{prop}

\paragraph{Coordinate expression:}
In coordinates we have
$$
\d E_\Lag = \derpar{E_\Lag}{q_i^A}\d q_i^A + \derpar{E_\Lag}{q_j^A}\d q_j^A
\quad ; \quad (0 \leqslant i \leqslant k-1,\ k \leqslant j \leqslant 2k-1) \ ,
$$
and taking the pull-back of this $1$-form by the section
$s(q_i^A) = (q_i^A,s_j^A)$, we obtain
\begin{align*}
s^*(\d E_\Lag)
= \left( \derpar{E_\Lag}{q_i^A} + \derpar{E_\Lag}{q_j^B}\derpar{s_j^B}{q_i^A} \right) \d q_i^A \, .
\end{align*}
Hence, the condition $\d(s^*E_\Lag) = 0$ in Proposition \ref{prop:LagHJEquiv}
is equivalent to the following $kn$ partial differential equations (on $\textnormal{Im}(s)$)
\begin{equation}\label{eqn:LagHJLocal}
\derpar{E_\Lag}{q_i^A} + \derpar{E_\Lag}{q_j^B}\derpar{s_j^B}{q_i^A} = 0
\end{equation}
Therefore, a section $s \in \Gamma(\rho^{2k-1}_{k-1})$ given locally by
$s(q_i^A) = (q_i^A,s_j^A(q_i^A))$ is a solution to the $k$th-order Lagrangian Hamilton-Jacobi problem
if, and only if, the local functions $s_j^A$ satisfy the system of $2kn$ partial differential equations
given by \eqref{eqn:GenLagHJLocal} and \eqref{eqn:LagHJLocal}.
Note that these $2kn$ partial differential equations may not be $\Cinfty(U)$-linearly independent.

In addition to the local equations for the section $s \in \Gamma(\rho^{2k-1}_{k-1})$,
we can state the equations for the characteristic Hamilton-Jacobi function.
These equations are a generalization to higher-order systems
of the classical Lagrangian Hamilon-Jacobi equations:
\begin{equation}\label{eqn:LagHJCharFunctLocalClassic}
\derpar{S}{q^A} = \derpar{\Lag}{v^A}(q^A,X^A)
\end{equation}
where $S \in \Cinfty(Q)$ is the characteristic function, $\Lag \in \Cinfty(\Tan Q)$ is the first-order Lagrangian,
$(q^A,v^A)$ are the natural coordinates in $\Tan Q$ and $X^A$ are the component functions
of the vector field $X \in \vf(Q)$ solution to the Lagrangian Hamilton-Jacobi problem
\cite{art:Carinena_Gracia_Marmo_Martinez_Munoz_Roman06}.

As $\omega_\Lag = -\d\theta_\Lag$, it is clear that
$s^*\omega_\Lag = 0$ if, and only if, $s^*(\d\theta_\Lag) = \d(s^*\theta_\Lag) = 0$;
that is, $s^*\theta_\Lag \in \df^{1}(\Tan^{k-1}Q)$
is a closed $1$-form. Using Poincar\'{e}'s Lemma, $s^*\theta_\Lag$ is locally exact, and thus
there exists $W \in \Cinfty(U)$, with $U \subseteq \Tan^{k-1}Q$
an open set, such that $s^*\theta_\Lag\vert_U = \d W$. In coordinates,
$$
\theta_\Lag = \sum_{i=0}^{k-1} \sum_{l=0}^{k-i-1} (-1)^ld_T^l\left( \derpar{\Lag}{q_{i+l+1}^A} \right)\d q_{i}^A \, ,
$$
and thus the local expression of $s^*\theta_\Lag$ is
$$
s^*\theta_\Lag = \sum_{i=0}^{k-1} \sum_{l=0}^{k-i-1} \restric{(-1)^ld_T^l\left( \derpar{\Lag}{q_{i+1+l}^A} \right)}{{\rm Im}(s)} \d q_i^A \\
$$
Hence, from the identity $s^*\theta_\Lag = \d W$ we obtain
\begin{equation}\label{eqn:LagHJCharFunctLocal}
\derpar{W}{q_i^A} =
\sum_{l=0}^{k-i-1} \restric{(-1)^ld_T^l\left( \derpar{\Lag}{q_{i+1+l}^A} \right)}{{\rm Im}(s)} \, ,
\end{equation}
which is a system of $kn$ partial differential equations for $W$ that
clearly generalizes equations \eqref{eqn:LagHJCharFunctLocalClassic}
to higher-order systems.

\subsection{Complete solutions}
\label{sec:LagHJCompleteSolutions}

In the above Sections we have stated the $k$th-order Hamilton-Jacobi problem in the Lagrangian formalism,
and a section $s \in \Gamma(\rho^{2k-1}_{k-1})$ solution to this problem gives a particular solution
to the dynamical equation \eqref{eqn:LagDynEq}. Nevertheless, this is not a complete solution
to the system, since only the integral curves of $X_\Lag$ with initial conditions in $\textnormal{Im}(s)$
can be recovered from the solution to the Hamilton-Jacobi problem. Hence, in order to obtain
a complete solution to the problem, we need to foliate the phase space $\Tan^{2k-1}Q$
in such a way that
every leaf is the image set of a section solution to the $k$th-order Lagrangian Hamilton-Jacobi problem.
The precise definition is:

\begin{definition}
A \textnormal{complete solution to the $k$th-order Lagrangian Hamilton-Jacobi problem} is a local diffeomorphism
$\Phi \colon U \times \Tan^{k-1}Q \to \Tan^{2k-1}Q$, with $U \subseteq \R^{kn}$
an open set, such that for every $\lambda \in U$, the map
$s_\lambda(\bullet) \equiv \Phi(\lambda,\bullet) \colon \Tan^{k-1}Q \to \Tan^{2k-1}Q$ is a solution
to the $k$th-order Lagrangian Hamilton-Jacobi problem.
\end{definition}

\noindent\textbf{Remark:}
Usually, it is the set of maps $\{ s_\lambda \mid \lambda \in U\}$ which is called a
complete solution of the $k$th-order Lagrangian Hamilton-Jacobi problem, instead of the
map $\Phi$. Both definitions are clearly equivalent.

It follows from this last definition that a complete solution provides $\Tan^{2k-1}Q$
with a foliation transverse to the fibers, and that every leaf of this foliation
has dimension $kn$ and is invariant by the Lagrangian vector field $X_\Lag$.

Let $\Phi$ be a complete solution, and we consider the family of vector fields
$$
\left\{ X_\lambda = \Tan\rho^{2k-1}_{k-1} \circ X_\Lag \circ s_\lambda \in \vf(\Tan^{k-1}Q)
\ ; \ \lambda \in U \subseteq \R^{kn} \right\} \, ,
$$
where $s_\lambda \equiv \Phi(\lambda,\bullet)$. Then, the integral curves of $X_\lambda$,
for different $\lambda \in U$, will provide all the integral curves of the Lagrangian
vector field $X_\Lag$. That is, if $\bar{y} \in \Tan^{2k-1}Q$, then there exists
$\lambda_o \in U$ such that if $p_o = \rho^{2k-1}_{k-1}(\bar{y})$, then
$s_{\lambda_o}(p_o) = \bar{y}$, and the integral curve
of $X_{\lambda_o}$ through $p_o$, lifted to $\Tan^{2k-1}Q$ by $s_{\lambda_o}$,
gives the integral curve of $X_\Lag$ through $\bar{y}$.

\section{The Hamilton-Jacobi problem in the Hamiltonian formalism}
\label{HamHamProb}

Let $Q$ be a $n$-dimensional smooth manifold modeling the
configuration space of a $k$th-order autonomous dynamical system
with $n$ degrees of freedom, and let $h \in
\Cinfty(\Tan^*(\Tan^{k-1}Q))$ be a Hamiltonian function containing
the dynamical information for this system. Using the canonical Liouville forms
of the cotangent bundle, namely
$\theta_{k-1} = p_{A}^{i}\d q_{i}^{A} \in
\df^{1}(\Tan^*(\Tan^{k-1}Q))$ and $\omega_{k-1} = \d q_{i}^{A}\wedge
\d p_{A}^{i} \in \df^{2}(\Tan^*(\Tan^{k-1}Q))$, where
$(q_{i}^{A},p_{i}^{A})$ with $1\leqslant A\leqslant n,$ $0\leqslant i\leqslant k-1$ are
canonical coordinates on $\Tan^{*}(\Tan^{k-1}Q)$, we can state the
dynamical equation for this Hamiltonian system,
\begin{equation}\label{eqn:HamDynEq}
\inn(X_h)\,\omega_{k-1} = \d h \, ,
\end{equation}
which has a unique solution $X_h \in \vf(\Tan^*(\Tan^{k-1}Q))$
due to the fact that $\omega_{k-1}$ is non-degenerate,
regardless of the Hamiltonian function provided.

As the formalism is developed in the cotangent bundle $\Tan^*(\Tan^{k-1}Q)$,
the statement of the Hamiltonian Hamilton-Jacobi theory for higher-order
systems follows the same pattern as in the first-order case
(see \cite{art:Carinena_Gracia_Marmo_Martinez_Munoz_Roman06}).
Next we detail the main results.

\subsection{The generalized Hamiltonian Hamilton-Jacobi problem}
\label{sec:GenHamHJProb}

\begin{definition}\label{def:GenHamHJProblem}
The {\rm generalized $k$th-order Hamiltonian Hamilton-Jacobi problem}
consists in finding a $1$-form $\alpha \in \df^{1}(\Tan^{k-1}Q)$
and a vector field $X \in \vf(\Tan^{k-1}Q)$ such that,
if $\gamma \colon \R \to \Tan^{k-1}Q$ is an integral curve of $X$,
then $\alpha \circ \gamma \colon \R \to \Tan^*(\Tan^{k-1}Q)$ is an integral
curve of $X_h$; that is,
\begin{equation}\label{eqn:GenHamHJDef}
X \circ \gamma = \dot{\gamma} \Longrightarrow
X_h \circ (\alpha \circ \gamma) = \dot{\overline{\alpha \circ \gamma}} \, .
\end{equation}
\end{definition}

\begin{prop}\label{prop:GenHamHJRelatedVF}
The pair $(\alpha,X) \in \df^{1}(\Tan^{k-1}Q) \times \vf(\Tan^{k-1}Q)$
safisfies the condition \eqref{eqn:GenHamHJDef} if, and only if,
$X$ and $X_h$ are $\alpha$-related, that is, $X_h \circ \alpha = \Tan\alpha \circ X$.
\end{prop}
\begin{proof}
The proof is like in the Lagrangian case and is the same as the one
of Prop. 5 in
\cite{art:Carinena_Gracia_Marmo_Martinez_Munoz_Roman06}.
\end{proof}

Now, from Proposition \ref{prop:GenHamHJRelatedVF},
composing both sides of the equality
$X_h \circ \alpha = \Tan\alpha \circ X$
with $\Tan\pi_{\Tan^{k-1}Q}$, and bearing in mind
that $\alpha \in \df^{1}(\Tan^{k-1}Q) = \Gamma(\pi_{\Tan^{k-1}Q})$, we obtain the
following result:

\begin{corol}
If $(\alpha,X)$ satisfies condition \eqref{eqn:GenHamHJDef}, then
$X = \Tan\pi_{\Tan^{k-1}Q} \circ X_h \circ \alpha$.
\end{corol}

Hence, the vector field $X \in \vf(\Tan^{k-1}Q)$ is completely
determined by the $1$-form $\alpha$, and it is called the
\textsl{vector field associated to $\alpha$}. The
following diagram illustrates the situation
$$
\xymatrix{
\Tan(\Tan^{k-1}Q) \ar@/_1.5pc/[rrr]_{\Tan \alpha} & \ & \ &
 \ar[lll]_-{\Tan\pi_{\Tan^{k-1}Q}} \Tan(\Tan^*(\Tan^{k-1}Q)) \\
\ & \ & \ & \ \\
\Tan^{k-1}Q \ar[uu]^{X} \ar@/_1.5pc/[rrr]_{\alpha} & \ & \ &
 \ar[lll]_{\pi_{\Tan^{k-1}Q}} \Tan^*(\Tan^{k-1}Q) \ar[uu]_{X_h}
}
$$

Since the vector field $X$ is completely determined by the $1$-form $\alpha$, the problem
of finding a pair
$(\alpha,X) \in \df^{1}(\Tan^{k-1}Q) \times \vf(\Tan^{k-1}Q)$
that satisfies the condition \eqref{eqn:GenHamHJDef} is equivalent to the problem of finding a $1$-form
$\alpha \in \df^{1}(\Tan^{k-1}Q)$ satisfying the same condition
with the associated vector field
$\Tan\pi_{\Tan^{k-1}Q} \circ X_h \circ \alpha \in \vf(\Tan^{k-1}Q)$.
Hence, we can define:

\begin{definition}
A \textnormal{solution to the generalized $k$th-order Hamiltonian Hamilton-Jacobi problem}
for $X_h$ is a $1$-form $\alpha \in \df^{1}(\Tan^{k-1}Q)$
such that if $\gamma \colon \R \to \Tan^{k-1}Q$ is an integral curve of
$X = \Tan\pi_{\Tan^{k-1}Q} \circ X_h \circ \alpha$,
then $\alpha \circ \gamma \colon \R \to \Tan^*(\Tan^{k-1}Q)$ is an integral curve of
$X_h$; that is,
$$
\Tan\pi_{\Tan^{k-1}Q} \circ X_h \circ \alpha \circ \gamma = \dot \gamma
\Longrightarrow X_h \circ (\alpha \circ \gamma) = \dot{\overline{\alpha \circ \gamma}} \, .
$$
\end{definition}

\begin{prop}\label{prop:GenHamHJEquiv}
The following conditions on a $1$-form
$\alpha \in \df^{1}(\Tan^{k-1}Q)$ are equivalent.
\begin{enumerate}
\item The form $\alpha$ is a solution to the generalized $k$th-order Hamiltonian Hamilton-Jacobi problem.
\item The submanifold $\textnormal{Im}(\alpha) \hookrightarrow \Tan^*(\Tan^{k-1}Q)$
is invariant under the flow of the Hamiltonian vector field $X_h$
(that is, $X_h$ is tangent to the submanifold $\textnormal{Im}(\alpha)$).
\item The form $\alpha$ satisfies the equation
\begin{equation*}
\inn(X)\d\alpha = -\d(\alpha^*h) \, ,
\end{equation*}
where $X = \Tan\pi_{\Tan^{k-1}Q} \circ X_h \circ \alpha$
is the vector field associated to $\alpha$.
\end{enumerate}
\end{prop}
\begin{proof}
The proof follows exactly the same pattern as in Proposition \ref{prop:GenLagHJEquiv}, taking into
account that from the properties of the tautological form $\theta_{k-1} \in \df^{1}(\Tan^*(\Tan^{k-1}Q))$
of the cotangent bundle, that is, we have $\alpha^*\theta_{k-1} = \alpha$ for every
$\alpha \in \df^{1}(\Tan^{k-1}Q)$.
Hence, taking the pull-back of the dynamical equation \eqref{eqn:HamDynEq}
by $\alpha$ we obtain
$$
\inn(X)\d\alpha = -\d(\alpha^*h) \, ,
$$
because we have
\begin{equation}\label{eqn:PullBackSymplecticFormByAlpha}
\alpha^*\omega_{k-1} = \alpha^*(-\d\theta_{k-1}) = -\d(\alpha^*\theta_{k-1}) = -\d\alpha \, .
\end{equation}
\end{proof}

\paragraph{Coordinate expression:}

Let $(q_0^A)$ be local coordinates in $Q$ and
$(q_0^A,\ldots,q_{k-1}^A)$ the induced natural coordinates in $\Tan^{k-1}Q$.
Then, $(q_0^A,\ldots,q_{k-1}^A,p_A^{0},\ldots,p_A^{k-1}) \equiv (q_i^A,p_A^i)$ are
natural coordinates in $\Tan^*(\Tan^{k-1}Q)$, which are also the adapted coordinates
to the $\pi_{\Tan^{k-1}Q}$-bundle structure.
Hence, a $1$-form $\alpha \in \df^{1}(\Tan^{k-1}Q)$ is given locally by
$\alpha(q_i^A) = (q_i^A,\alpha_A^i) = \alpha_A^i \d q_i^A$,
where $\alpha_A^i$ are local smooth functions in $\Tan^{k-1}Q$.

If $\alpha \in \df^{1}(\Tan^{k-1}Q)$ is a solution to the
generalized $k$th-order Hamiltonian Hamilton-Jacobi problem,
then by Proposition \ref{prop:GenHamHJEquiv}
this is equivalent to require the Hamiltonian vector field $X_h \in \vf(\Tan^*(\Tan^{k-1}Q))$
to be tangent to the submanifold $\textnormal{Im}(\alpha) \hookrightarrow \Tan^*(\Tan^{k-1}Q)$.
This submanifold is locally defined by the constraints $p_A^l - \alpha_A^l = 0$.
Thus, we must require $\Lie(X_h)(p_A^i-\alpha_A^i) \equiv X_h(p_A^i-\alpha_A^i) = 0$  (on $\textnormal{Im}(\alpha)$).
The Hamiltonian vector field $X_h$ is locally given by \cite{book:DeLeon_Rodrigues85,art:Prieto_Roman11}
\begin{equation}\label{eqn:HamVF}
X_h = \derpar{h}{p_A^i}\derpar{}{q_i^A} - \derpar{h}{q_i^A}\derpar{}{p_A^i} \, .
\end{equation}
Hence, the conditions $\restric{X_h(p_A^i-\alpha_A^i)}{\textnormal{Im}(\alpha)} = 0$
give the equations
\begin{equation}\label{eqn:GenHamHJLocal}
-\derpar{h}{q_i^A} - \derpar{h}{p_B^j}\derpar{\alpha_A^i}{q_j^B} = 0  \quad , \quad
\mbox{(on ${\rm Im}(\alpha)$)}\ .
\end{equation}
This is a system of $kn$ partial differential equations with $kn$
unknown functions $\alpha_A^i$
which must be verified by every $1$-form $\alpha \in \df^{1}(\Tan^{k-1}Q)$ solution
to the generalized $k$th-order Hamiltonian Hamilton-Jacobi.

\subsection{The Hamiltonian Hamilton-Jacobi problem}
\label{sec:HamHJProblem}

As in the Lagrangian setting, it is convenient to consider a less general problem requiring some
additional conditions to the $1$-form $\alpha \in \df^{1}(\Tan^{k-1}Q)$.
Observe that from \eqref{eqn:PullBackSymplecticFormByAlpha} the
condition $\alpha^*\omega_{k-1} = 0$ is equivalent to $\d\alpha = 0$; that is,
$\alpha$ is a closed $1$-form in $\Tan^{k-1}Q$. Therefore:

\begin{definition}\label{def:HamHJProblemClassic}
The \textnormal{$k$th-order Hamiltonian Hamilton-Jacobi problem}
consists in finding closed $1$-forms $\alpha \in \df^{1}(\Tan^{k-1}Q)$
solution to the generalized Hamiltonian Hamilton-Jacobi problem.
Such a form is called a \textnormal{solution to the $k$th-order Hamiltonian Hamilton-Jacobi problem}.
\end{definition}

A straightforward consequence of Proposition \ref{prop:GenHamHJEquiv} is the
following result:

\begin{prop}\label{prop:HamHJEquiv}
Let $\alpha \in \df^{1}(\Tan^{k-1}Q)$ be a closed $1$-form.
The following assertions are equivalent:
\begin{enumerate}
\item The $1$-form $\alpha$ is a solution to the $k$th-order Hamiltonian Hamilton-Jacobi problem.
\item $\d(\alpha^*h) = 0$.
\item $\textnormal{Im}(\alpha)$ is a Lagrangian submanifold of $\Tan^*(\Tan^{k-1}Q)$ invariant by $X_h$.
\item The integral curves of $X_h$ with initial conditions in $\textnormal{Im}(\alpha)$ project
onto the integral curves of $X = \Tan\pi_{\Tan^{k-1}Q} \circ X_h \circ \alpha$.
\end{enumerate}
\end{prop}

\paragraph{Coordinate expression:}

In coordinates we have
$$
\d h = \derpar{h}{q_i^A}\d q_i^A + \derpar{h}{p_A^i} \d p_A^i \quad ; \quad
(0 \leqslant i \leqslant k-1) \, ,
$$
and taking the pull-back of $\d h$ by the $1$-form $\alpha = \alpha_A^i \d q_i^A$, we have
$$
\alpha^*(\d h) = \left( \derpar{h}{q_i^A} +
\derpar{h}{p_B^j}\derpar{\alpha_B^j}{q_i^A} \right) \d q_i^A \, .
$$
Hence, the condition $\d(\alpha^*h) = 0$ in Proposition \ref{prop:HamHJEquiv} holds if, and only if,
the following $kn$ partial differential equations hold
\begin{equation}\label{eqn:HamHJLocal}
\derpar{h}{q_i^A} + \derpar{h}{p_B^j}\derpar{\alpha_B^j}{q_i^A} = 0
\end{equation}

Equivalently, we can require the $1$-form $\alpha \in \df^{1}(\Tan^{k-1}Q)$
to be closed, that is, $\d\alpha = 0$. Locally, this condition reads
\begin{equation}\label{eqn:HamHJLocalClosedForm}
\derpar{\alpha_A^i}{q_j^B} - \derpar{\alpha_B^j}{q_i^A} = 0 \ ,
\, \mbox{with } A \neq B \mbox{ or } i \neq j \, .
\end{equation}

Therefore, a $1$-form
$\alpha \in \df^{1}(\Tan^{k-1}Q)$ given locally by
$\alpha = \alpha_A^i\d q_i^A$ is a solution to the $k$th-Hamiltonian Hamilton-Jacobi
problem if, and only if, the local functions $\alpha_A^i$ satisfy the
$2kn$ partial differential equations given by \eqref{eqn:GenHamHJLocal} and
\eqref{eqn:HamHJLocal}, or equivalently \eqref{eqn:GenHamHJLocal} and
\eqref{eqn:HamHJLocalClosedForm}.
Observe that these $2kn$ partial differential equations may not be $\Cinfty(U)$-linearly independent.

In addition to the local equations for the $1$-form $\alpha \in \df^{1}(\Tan^{k-1}Q)$,
in this particular situation we can give the equation for the characteristic
Hamilton-Jacobi function. This equation is a generalization
to higher-order systems of the classical Hamiltonian Hamilton-Jacobi equation
\begin{equation}\label{eqn:HamHJClassicEquation}
h\left( q^A,\derpar{S}{q^A} \right) = E
\end{equation}
where $E \in \R$ is a constant, $S \in \Cinfty(Q)$ is the characteristic function and
$h \in \Cinfty(\Tan^*Q)$ is the Hamiltonian function.

As $\alpha \in \df^{1}(\Tan^{k-1}Q)$ is closed, by Poincar{\'e}'s Lemma
there exists a function $W \in \Cinfty(U)$, with $U \subseteq \Tan^{k-1}Q$ an open set,
such that $\alpha = \d W$. In coordinates
the condition $\alpha = \d W$ gives the following
$kn$ partial differential equations for $W$
\begin{equation}\label{eqn:HamHJCharFunctLocal}
\derpar{W}{q_i^A} = \alpha_A^i \, .
\end{equation}
Finally, as $\alpha^*h = h(q_i^A,\alpha_A^i) = h\left( q_i^A, \derpar{W}{q_i^A} \right)$,
the condition $\alpha^*h$ being locally constant gives
\begin{equation}\label{eqn:HamHJEquation}
h\left( q_i^A, \derpar{W}{q_i^A} \right) = E \, ,
\end{equation}
where $E \in \R$ is a local constant. This equation clearly generalizes the equation
\eqref{eqn:HamHJClassicEquation} to higher-order systems.

\subsection{Complete solutions}
\label{sec:HamHJCompleteSolutions}

The concept of {\sl complete solution} is defined
in an analogous way as in Section \ref{sec:LagHJCompleteSolutions}.

\begin{definition}
A \textnormal{complete solution to the $k$th-order Hamiltonian Hamilton-Jacobi problem} is
is a local diffeomorphism $\Phi \colon U \times \Tan^{k-1}Q \to \Tan^*(\Tan^{k-1}Q)$,
where $U \subseteq \R^{kn}$ is an open set,
such that, for every $\lambda \in U$, the map
 $\alpha_\lambda(\bullet)\equiv\Phi(\lambda,\bullet) \colon \Tan^{k-1}Q \to \Tan^*(\Tan^{k-1}Q)$
is a solution to the $k$th-order Hamiltonian Hamilton-Jacobi problem.

Then, the set $\{ \alpha_\lambda \mid \lambda \in U\}$ is also called a
\textnormal{complete solution to the $k$th-order Hamiltonian Hamilton-Jacobi problem}.
\end{definition}

It follows from the definition that a complete solution endows $\Tan^*(\Tan^{k-1}Q)$
with a foliation transverse to the fibers, and that the Hamiltonian vector field $X_h$
is tangent to the leaves.

Let $\{ \alpha_\lambda \mid \lambda \in U \}$ be a complete solution, and we consider the
set of associated vector fields
$$
\left\{ X_\lambda = \Tan\pi_{\Tan^{k-1}Q} \circ X_h \circ \alpha_\lambda \in \vf(\Tan^{k-1}Q)
\ ; \ \lambda \in U \subseteq \R^{kn} \right\} \, .
$$
Then, the integral curves of $X_\lambda$,
for different $\lambda \in U$, will provide all the integral curves of the Hamiltonian
vector field $X_h$. That is, if $\beta \in \Tan^*(\Tan^{k-1}Q)$, then there exists
$\lambda_o \in U$ such that if $p_o = \pi_{\Tan^{k-1}Q}(\beta)$,
then $\alpha_{\lambda_o}(p_o) = \beta$, and the integral curve
of $X_{\lambda_o}$ through $p_o$, lifted  to $\Tan^*(\Tan^{k-1}Q)$ by $\alpha_{\lambda_o}$,
gives the integral curve of $X_h$ through $\beta$.

Let us assume that $\Phi$ is a global diffeomorphism for simplicity. Then, given
$\lambda = (\lambda_i^A) \in \R^{kn}$, $0 \leqslant i \leqslant k-1$, $1 \leqslant A \leqslant n$,
we consider the functions $f_{j}^{B} \colon \Tan^{*}(\Tan^{k-1}Q) \to \R$,
$0 \leqslant j \leqslant k-1$, $1 \leqslant B \leqslant n$, given by
$$
f_{j}^{B} = \pr_{j}^{B} \circ \p_1 \circ \Phi^{-1} \, ,
$$
where $\p_1 \colon \R^{kn} \times \Tan^{k-1}Q \to \R^{kn}$ is the projection onto the first factor
and $\pr_{j}^{B} \colon \R^{kn} \to \R$ is given by $\pr_{j}^{B} = \pr^{B} \circ \pr_{j}$,
where $\pr^B$ and $\pr_i$ are the natural projections
$$
\begin{array}{rcl}
\pr_j \colon \R^{kn} & \longrightarrow & \R^{n} \\
(x_i^A) & \longmapsto & (x_{j}^1,\ldots,x_j^n)
\end{array} \quad ; \quad
\begin{array}{rcl}
\pr^{B} \colon \R^{n} & \longrightarrow & \R \\
(x^1,\ldots,x^n) & \longmapsto & x^{B}.
\end{array}
$$
Therefore, $f_{j}^{B}(\alpha_{\lambda}(q_i^A)) = (\pr_{j}^{B} \circ \p_{1} \circ \Phi^{-1} \circ \Phi)(\lambda_i^A,q_i^A) = \lambda_{j}^{B}$.

\begin{prop}\label{involution}
The functions $f_{j}^{B}$, $0 \leqslant j \leqslant k-1,$ $1 \leqslant B \leqslant n$ are in involution.
\end{prop}
\begin{proof}
Let $\beta \in \Tan^{*}(\Tan^{k-1}Q)$. We will show that
$\{f_{i}^{j},f_{a}^{b}\}(\beta)=0$.

Since $\Phi$ is a complete solution, for every $\beta \in \Tan^*(\Tan^{k-1}Q)$
there exists $\lambda \in \R^{kn}$
such that $\alpha_\lambda(\pi_{\Tan^{k-1}Q}(\beta)) = \Phi(\lambda,\pi_{\Tan^{k-1}Q}(\beta)) = \beta$.
Then we have
$$
f_j^B(\beta) = (f_j^B \circ \alpha_\lambda)(\pi_{\Tan^{k-1}Q}(\beta))
= (\pr_{j}^{B} \circ \p_1 \circ \Phi^{-1} \circ \Phi)(\lambda,\pi_{\Tan^{k-1}Q}(\beta))
= \lambda_j^B \, ,
$$
that is, $f_j^B \circ \alpha_\lambda = (f_j^B \circ \Phi)(\lambda,\bullet) \colon \Tan^{k-1}Q \to \R$
is constant for every $\lambda \in \R^{kn}$. Therefore, we have
$$
\restric{\d f_{j}^{B}}{\Tan \textnormal{Im}(\alpha_{\lambda})}=0 \, .
$$

Now, since $\Phi$ is a complete solution, we have that
$\alpha_{\lambda} = \Phi(\lambda,\bullet)$ is a solution to the $k$th-order Hamiltonian
Hamilton-Jacobi problem. Therefore, from Prop. \ref{prop:HamHJEquiv},
$\textnormal{Im}(\Phi_{\lambda})$ is a Lagrangian submanifold of
$(\Tan^{*}(\Tan^{k-1}Q),\omega_{k-1})$, and then
$$
\left(\Tan \textnormal{Im}(\alpha_{\lambda})\right)^{\bot} = \Tan \textnormal{Im}(\alpha_{\lambda}) \, ,
$$
where $\left(\Tan \textnormal{Im}(\alpha_{\lambda})\right)^{\bot}$ denotes the
$\omega_{k-1}$-orthogonal of $\Tan \textnormal{Im}(\alpha_{\lambda})$.

From this, the result follows from the definition of the induced Poisson bracket,
which is
$$
\{ f_{j}^{B},f_{l}^{C}\}(\beta) = \omega_{k-1}(X_{f_{j}^{B}},X_{f_{l}^{C}})(\beta) \, ,
$$
and the facts that $\omega_{k-1}$ is symplectic,
$\d f_{j}^{B} \in \left(\Tan \textnormal{Im}(\alpha_{\lambda})\right)^{\bot} = \Tan \textnormal{Im}(\alpha_{\lambda})$,
and that there exists a unique vector field
$X_{f_{j}^{B}} \in \vf(\Tan^*(\Tan^{k-1}Q))$ satisfying
$\inn(X_{f_{j}^{B}})\omega_{k-1} = \d f_{j}^{B}$.
\end{proof}

\subsection{Relation with the Lagrangian formulation}

Up to this point we have stated both the Lagrangian and Hamiltonian
Hamilton-Jacobi problems for higher-order autonomous systems. Now, we
establish a relation between the solutions of the Hamilton-Jacobi problem in
both formulations. In particular, we show that
there exists a bijection between the set of solutions of the (generalized)
$k$th-order Lagrangian Hamilton-Jacobi problem and the set of solutions of the
(generalized) $k$th-order Hamiltonian Hamilton-Jacobi problem, given by the
Legendre-Ostrogradsky map.

\begin{definition} Let $(\Tan^{2k-1}Q,\Lag)$ be a
Lagrangian system. The {\rm Legendre-Ostrogradsky map} (or {\rm
generalized Legendre map\/}) associated to $\Lag$ is the map $\Leg
\colon \Tan^{2k-1}Q \to \Tan^*(\Tan^{k-1}Q)$ defined as follows:
 for every $u \in \Tan(\Tan^{2k-1}Q)$,
$$
\theta_\Lag(u) = \left\langle \Tan \rho_{k-1}^{2k-1}(u), \Leg(\tau_{\Tan^{2k-1}Q}(u)) \right\rangle
$$
\end{definition}

This map verifies that $\pi_{\Tan^{k-1}Q} \circ \Leg
= \rho^{2k-1}_{k-1}$, that is, it is a bundle map over $\Tan^{k-1}Q$.
Furthermore, we have that $\Leg^*\theta_{k-1} = \theta_\Lag$ and
$\Leg^*\omega_{k-1} = \omega_\Lag$.
Moreover, we have that $\Lag \in \Cinfty(\Tan^{k}Q)$ is a regular Lagrangian
if, and only if, $\Leg \colon \Tan^{2k-1}Q \to \Tan^*(\Tan^{k-1}Q)$ is a local
diffeomorphism, and $\Lag$ is said to be \textsl{hyperregular} if $\Leg$ is a global
diffeomorphism.

\noindent\textbf{Remark:} Observe that if $\Lag$ is
hyperregular, then $\Leg$ is a symplectomorphism and therefore the
symplectic structures are in correspondence. Therefore, the induced
Poisson brackets also are in correspondence and we have the
analogous result of \eqref{involution} in the Lagrangian formalism,
where the Poisson bracket is determined by the Poincar\'e-Cartan
$2$-form $\omega_{\Lag}$ as $\{f,g\} = \omega_{\Lag}(X_f,X_g).$

Given a local natural chart in $\Tan^{2k-1}Q$, we
can define the following local functions
$$
\hat p^{r-1}_A = \sum_{i=0}^{k-r}(-1)^i d_T^i\left(\derpar{\Lag}{q_{r+i}^A}\right) \ .
$$
Thus, bearing in mind the local expression of the form
$\theta_\Lag$, we can write $\theta_\Lag = \sum_{r=1}^k \hat
p^{r-1}_A\d q_{r-1}^A$, and we obtain that the expression in natural
coordinates of the map $\Leg$ is
$$
\Leg\left(q_i^A,q_j^A\right) = \left(q_i^A,p^i_A\right) \ , \
\mbox{\rm with $p^i_A\circ\Leg=\hat p^i_A$} \ .
$$

Now we establish the relation theorem. First we need the
following technical result:

\begin{lemma}\label{lemma:TechLemma}
Let $E_1 \stackrel{\pi_1}{\longrightarrow} M$ and $E_2 \stackrel{\pi_2}{\longrightarrow} M$
be two fiber bundles, $F \colon E_1 \to E_2$ a
fiber bundle morphism, and two $F$-related vector fields $X_1 \in \vf(E_1)$ and $X_2 \in \vf(E_2)$.
 If $s_1 \in \Gamma(\pi_1)$ is a section of $\pi_1$ and we define
a section of $\pi_2$ as $s_2 = F \circ s_1 \in \Gamma(\pi_2)$, then
$$
\Tan\pi_1 \circ X_1 \circ s_1 = \Tan\pi_2 \circ X_2 \circ s_2 \in \vf(M) \, .
$$
\end{lemma}
\begin{proof}
As $F \colon E_1 \to E_2$ is a fiber bundle morphism (that is,
$\pi_1 = \pi_2 \circ F$), and $X_1$ and $X_2$ are $F$-related (that
is, $\Tan F \circ X_1 = X_2 \circ F$), we have the following
commutative diagram
$$
\xymatrix{
\Tan E_1 \ar[rr]^{\Tan F} &  \ & \Tan E_2 \\
E_1 \ar[dr]_{\pi_1} \ar[u]^{X_1} \ar[rr]^{F} & \ & E_2 \ar[dl]^{\pi_2} \ar[u]_{X_2} \\
\ & M & \
}
$$
Then we have
\begin{align*}
\Tan\pi_1 \circ X_1 \circ s_1 &= \Tan(\pi_2 \circ F) \circ X_1 \circ s_1
= \Tan\pi_2 \circ \Tan F \circ X_1 \circ s_1 \\
&= \Tan\pi_2 \circ X_2 \circ F \circ s_1
= \Tan\pi_2 \circ X_2 \circ s_2 \, .
\end{align*}
\end{proof}

Then, the equivalence theorem is:

\begin{theorem}\label{thm:EquivalenceSolutionsLag&Ham}
Let $(\Tan^{2k-1}Q,\Lag)$ be a hyperregular Lagrangian system,
$(\Tan^*(\Tan^{k-1}Q),\omega_{k-1},h)$ its associated Hamiltonian system.
\begin{enumerate}
\item If $s \in \Gamma(\rho^{2k-1}_{k-1})$ is a solution to the (generalized) $k$th-order Lagrangian
Hamilton-Jacobi  problem, then the $1$-form
$\alpha = \Leg \circ s  \in \df^{1}(\Tan^{k-1}Q)$
is a solution to the (generalized) $k$th-order Hamiltonian Hamilton-Jacobi problem.
\item If $\alpha \in \df^{1}(\Tan^{k-1}Q)$ is a solution to
the (generalized) $k$th-order Hamiltonian Hamilton-Jacobi problem, then the section
$s  = \Leg^{-1} \circ \alpha \in \Gamma(\rho^{2k-1}_{k-1})$
is a solution to the (generalized) $k$th-order Lagrangian Hamilton-Jacobi problem.
\end{enumerate}
\end{theorem}
\begin{proof}
\begin{enumerate}
\item
Let $X = \Tan\rho^{2k-1}_{k-1} \circ X_\Lag \circ s$ and
$\bar{X} = \Tan\pi_{\Tan^{k-1}Q} \circ X_h \circ \alpha$ be the vector
fields associated to $s$ and $\alpha = \Leg \circ s$, respectively. From
Lemma \ref{lemma:TechLemma} we have $X = \bar{X}$, and hence
both vector fields are denoted by $X$.

Let $s$ be a solution to the generalized $k$th-order Lagrangian
Hamilton-Jacobi problem, and $\gamma \colon \R \to \Tan^{k-1}Q$
an integral curve of $X$. Therefore
\begin{align*}
X_h \circ (\alpha \circ \gamma) &= X_h \circ \Leg \circ s \circ \gamma
= \Tan\Leg \circ X_\Lag \circ s \circ \gamma \\
&= \Tan\Leg \circ \Tan s \circ X \circ \gamma
= \Tan(\Leg \circ s) \circ \dot\gamma \\
&= \Tan \alpha \circ \dot\gamma = \dot{\overline{\alpha \circ \gamma}} \ ,
\end{align*}
then $\alpha \circ \gamma$ is an integral curve of $X_h$
and thus $\alpha$ is a solution to the generalized $k$th-order Hamiltonian Hamilton-Jacobi problem.

Now, in addition, we require $s^*\omega_\Lag = 0$; that is, $s$
is a solution to the $k$th-order Lagrangian Hamilton-Jacobi problem. Then,
using \eqref{eqn:PullBackSymplecticFormByAlpha} we have
$$
\d\alpha = \alpha^*\omega_{k-1} = (\Leg \circ s)^*\omega_{k-1} = s^*(\Leg^*\omega_{k-1})
= s^*\omega_\Lag = 0 \, ,
$$
and hence $\alpha$ is a solution to the $k$th-order Hamiltonian Hamilton-Jacobi problem.
\item
The proof is analogous to the one for the item 1, but using $\Leg^{-1}$.
\end{enumerate}
\end{proof}

This result can be extended to complete solutions in a natural way.

Obviously, for regular but not hyperregular Lagrangian functions,
all these results hold only in the open sets where $\Leg$ is a diffeomorphism.

Theorem \ref{thm:EquivalenceSolutionsLag&Ham} allows us
to show that the vector field associated to a section solution to the (generalized)
Hamilton-Jacobi problem is a semispray of type $1$. First, we need the following technical result:

\begin{lemma}\label{lemma:TechLemma2}
Let $X \in \vf(\Tan^rQ)$ be a semispray of type $1$ on $\Tan^rQ$,
and $Y \in \vf(\Tan^sQ)$ ($s \leqslant r$)  which is $\rho^r_s$-related with $X$.
Then $Y$ is a semispray of type $1$ on $\Tan^sQ$.
\end{lemma}
\begin{proof}
Let $\gamma \colon \R \to \Tan^rQ$ be an integral curve of $X$.
Then, as $X$ is a semispray of type $1$,
there exists a curve $\phi \colon \R \to Q$ such that $j^{r}\phi = \gamma$.
Furthermore, as $X$ and $Y$ are $\rho^r_s$-related,
the curve $\rho^r_s \circ \gamma \colon \R \to \Tan^sQ$ is an integral curve of $Y$.
Hence, $\rho^r_s (j^{r}\phi) = j^s\phi$ is an integral curve of $Y$.

It remains to show that every integral curve of $Y$ is the projection to $\Tan^sQ$ via $\rho^r_s$ of an
integral curve of $X$, but this holds due to the fact that the vector fields are $\rho^r_s$-related and
$\rho^r_s$ is a surjective submersion. Therefore, $Y$ is a semispray of type $1$ in $\Tan^sQ$.
\end{proof}

\begin{prop}\label{prop:XSemispray1}
Let $(\Tan^{2k-1}Q,\Lag)$ be a hyperregular Lagrangian system,
and $(\Tan^*(\Tan^{k-1}Q),\omega_{k-1},h)$
the associated Hamiltonian system. Then, if $\alpha \in \df^{1}(\Tan^{k-1}Q)$ is a solution
to the $k$th-order Hamiltonian Hamilton-Jacobi problem, the vector field
$X = \Tan\pi_{\Tan^{k-1}Q} \circ X_h \circ \alpha$ is a semispray of type $1$
on $\Tan^{k-1}Q$.
\end{prop}
\begin{proof}
Let $s = \Leg^{-1} \circ \alpha \in \Gamma(\rho^{2k-1}_{k-1})$ be the section associated to $\alpha$.
Then, by Lemma \ref{lemma:TechLemma}, if
$X = \Tan\pi_{\Tan^{k-1}Q} \circ X_h \circ \alpha$
and $\bar{X} = \Tan\rho^{2k-1}_{k-1} \circ X_\Lag \circ s$
are the vector fields on $\Tan^{k-1}Q$ associated
to $\alpha$ and $s$ respectively, then $X = \bar{X} = \Tan\rho^{2k-1}_{k-1} \circ X_\Lag \circ s$.
Hence, as $X_\Lag \in \vf(\Tan^{2k-1}Q)$ is the Lagrangian vector field solution to the equation
\eqref{eqn:LagDynEq} and $\Lag \in \Cinfty(\Tan^{k}Q)$ is a hyperregular Lagrangian function, we have
that $X_\Lag$ is a semispray of type $1$ on $\Tan^{2k-1}Q$.
In particular, $X_\Lag \circ s$ is a semispray of type $1$ along $\rho^{2k-1}_{k-1}$ and,
by Lemma \ref{lemma:TechLemma2}, $X$ is a semispray of type $1$ on $\Tan^{k-1}Q$.
\end{proof}

As a consequence of Proposition \ref{prop:XSemispray1},
the generalized $k$th-order Hamilton-Jacobi problem
can be stated in the following way:

\begin{definition}
The \textnormal{generalized $k$th-order Lagrangian} (resp., \textnormal{Hamiltonian}) \textnormal{Hamilton-Jacobi problem}
consists in finding a section $s \in \Gamma(\rho^{2k-1}_{k-1})$
(resp., a $1$-form $\alpha \in \df^{1}(\Tan^{k-1}Q)$)
such that, if $\gamma \colon \R \to Q$ satisfies that $j^{k-1}\gamma$ is an integral curve of
$X = \Tan\rho^{2k-1}_{k-1} \circ X_\Lag \circ s$
(resp., $X = \Tan\pi_{\Tan^{k-1}Q} \circ X_h \circ \alpha$),
then $s \circ j^{k-1}\gamma \colon \R \to \Tan^{2k-1}Q$
(resp., $\alpha \circ j^{k-1}\gamma \colon \R \to \Tan^*(\Tan^{k-1}Q)$)
is an integral curve of $X_\Lag$ (resp., $X_h$).
\end{definition}

\section{Examples}
\label{sec:Examples}

\subsection{The end of a javelin}
\label{sec:Example1}

Let us consider the dynamical system that describes the motion of the end of a thrown javelin.
This gives rise
to a $3$-dimensional second-order dynamical system, which is a particular case of the problem
of determining the trajectory of a particle rotating about a translating center \cite{art:Constantelos84}.
Let $Q = \R^3$ be the manifold modeling the configuration space for this system with coordinates
$(q_0^1,q_0^2,q_0^3) = (q_0^A)$. Using the induced coordinates in $\Tan^2\R^3$,
the Lagrangian function for this system is
$$
\Lag(q_0^A,q_1^A,q_2^A) = \frac{1}{2} \sum_{A=1}^{3} \left((q_1^A)^2 - (q_2^A)^2\right) \ ,
$$
which is a regular Lagrangian function since the Hessian matrix of $\Lag$ with respect to the second-order
velocities is
$$
\left( \derpars{\Lag}{q_2^B}{q_2^A} \right) = \begin{pmatrix} 1 & 0 & 0 \\ 0 & 1 & 0 \\ 0 & 0 & 1 \end{pmatrix}.
$$


The Poincar\'{e}-Cartan forms $\theta_\Lag$ and $\omega_\Lag$, and the Lagrangian energy are
locally given by
$$
\begin{array}{c}
\displaystyle \theta_\Lag = \sum_{A=1}^{3} \left( (q_1^A + q_3^A)\d q_0^A - q_2^A \d q_1^A\right) \quad ; \quad
\omega_\Lag = \sum_{A=1}^{3} \left( \d q_0^A \wedge \d q_1^A + \d q_0^A \wedge \d q_3^A - \d q_1^A \wedge \d q_2^A \right) \\
\displaystyle E_\Lag = \frac{1}{2} \sum_{A=1}^{3} \left( (q_1^A)^2 + 2q_1^Aq_3^A - (q_2^A)^2 \right).
\end{array}
$$
Thus, the semispray of type $1$, $X_\Lag \in \vf(\Tan^3\R^3)$ solution to the dynamical equation
\eqref{eqn:LagDynEq} is
$$
X_\Lag = q_1^A\derpar{}{q_0^A} + q_2^A\derpar{}{q_1^A} + q_3^A\derpar{}{q_2^A} - q_2^A\derpar{}{q_3^A} \, .
$$

Consider the projection $\rho^3_1 \colon \Tan^3\R^3 \to \Tan\R^3$. From Proposition \ref{prop:GenLagHJEquiv}
we know that the generalized second-order Lagrangian Hamilton-Jacobi problem consists in finding sections
$s \in \Gamma(\rho^3_1)$ such that the Lagrangian vector field $X_\Lag$ is tangent to the submanifold
$\textnormal{Im}(s) \hookrightarrow \Tan^3\R^3$. Suppose that the section $s$ is given locally
by $s(q_0^A,q_1^A) = (q_0^A,q_1^A,s_2^A,s_3^A)$. As the submanifold $\textnormal{Im}(s)$
is defined locally by the constraint functions $q_2^A - s_2^A$ and $q_3^A - s_3^A$, then the tangency
condition gives the following system of $6$ partial differential equations for the component functions
of the section
$$
s_3^A - q_1^B\derpar{s_2^A}{q_0^B} - s_2^B\derpar{s_2^A}{q_1^B} = 0 \quad ; \quad
s_2^A + q_1^B\derpar{s_3^A}{q_0^B} + s_2^B\derpar{s_3^A}{q_1^B} = 0 \, .
$$

In order to obtain the equations of the second-order Lagrangian Hamilton-Jacobi problem, we require
in addition the section $s \in \Gamma(\rho^3_1)$ to satisfy the condition $\d(s^*E_\Lag) = 0$, or,
equivalently, $s^*\omega_\Lag = 0$. From the local expression of the Cartan $2$-form
$\omega_\Lag \in \df^{2}(\Tan^3\R^3)$ given above, taking the pull-back by the section
$s(q_0^A,q_1^A) = (q_0^A,q_1^A,s_2^A,s_3^A)$ we obtain
$$
s^*\omega_\Lag = \sum_{A=1}^{3} \left[ \d q_0^A \wedge \d q_1^A + \derpar{s_3^A}{q_0^B} \d q_0^A \wedge \d q_0^B
+ \left( \derpar{s_3^A}{q_1^B} + \derpar{s_2^B}{q_0^A} \right) \d q_0^A \wedge \d q_1^B
- \derpar{s_2^A}{q_1^B} \d q_1^A \wedge \d q_1^B \right].
$$
Hence, the condition $s^*\omega_\Lag = 0$ gives the following partial differential equations
$$
\derpar{s_3^A}{q_0^B} = \derpar{s_2^A}{q_1^B} = \derpar{s_3^A}{q_1^B} + \derpar{s_2^B}{q_0^A} = 0 \, , \ \mbox{if } A \neq B \, ,
\qquad \derpar{s_3^A}{q_1^A} + \derpar{s_2^A}{q_0^A} + 1 = 0 \, .
$$
Hence, the section $s \in \Gamma(\rho^3_1)$ is a solution to the second-order Lagrangian
Hamilton-Jacobi problem if the following system of partial differential equations hold
$$
\begin{array}{c}
\displaystyle s_3^A = q_1^B\derpar{s_2^A}{q_0^B} + s_2^A\derpar{s_2^A}{q_1^A} \quad , \quad
q_1^A\derpar{s_3^A}{q_0^A} + s_2^B\derpar{s_3^A}{q_1^B} + s_2^A = 0 \, , \\[10pt]
\displaystyle \derpar{s_3^A}{q_1^B} + \derpar{s_2^B}{q_0^A} = 0 \quad , \quad
\derpar{s_3^A}{q_1^A} + \derpar{s_2^A}{q_0^A} + 1 = 0 \, .
\end{array}
$$

Finally, we compute the equations for the generating function $W$. The pull-back of the Cartan $1$-form
$\theta_\Lag$ by the section $s$ gives in coordinates
$$
s^*\theta_\Lag = \sum_{A=1}^{3} \left( (q_1^A + s_3^A)\d q_0^A - s_2^A \d q_1^A\right)
$$
Hence, requiring $s^*\theta_\Lag = \d W$ for a local function $W$ defined in $\Tan Q$ we obtain
$$
\derpar{W}{q_0^A} = q_1^A + s_3^A \quad ; \quad \derpar{W}{q_1^A} = -s_2^A \, ,
$$
and thus from $\d(s^*E_\Lag) = 0$, we have $s^*E_\Lag = \mbox{const.}$, that is,
$$
 \sum_{A=1}^{3} \left( q_1^A\derpar{W}{q_0^A} - \frac{1}{2}\left((q_1^A)^2 + \left( \derpar{W}{q_1^A}\right)^2 \right)\right) = \mbox{const.}
$$

\noindent{\bf Remark:} This equation cannot be stated in the general
case, since we need to clear the higher-order velocities from the
previous equations. This calculation is easy for this particular
example, but it depends on the Lagrangian function provided in the
general case.


Now, to establish the Hamiltonian formalism for the Hamilton-Jacobi
problem, we consider natural coordinates $(q_0^A,q_1^A,p_A^0,p_A^1)$
on the cotagent bundle $\Tan^*(\Tan\R^3)$, then the
Legendre-Ostrogradsky map $\Leg \colon \Tan^3\R^3 \to
\Tan^*(\Tan\R^3)$ associated to the Lagrangian function $\Lag$ is
$$
\Leg^*q_0^A = q_0^A \quad ; \quad \Leg^*q_1^AÂ \quad ; \quad
\Leg^*p_A^0 = q_1^A + q_3^A \quad ; \quad \Leg^*p_A^1 = -q_2^A \, ,
$$
and the inverse map $\Leg^{-1} \colon \Tan^*(\Tan\R^3) \to \Tan^3\R^3$ is given by
$$
(\Leg^{-1})^*q_0^A = q_0^A \quad ; \quad (\Leg^{-1})^*q_1^A = q_1^A \quad ; \quad
(\Leg^{-1})^*q_2^A = -p_A^1 \quad ; \quad (\Leg^{-1})^*q_3^A = p_A^0 - q_1^A \, .
$$
From these coordinate expressions it is clear that the Legendre-Ostrogradsky map
is a global diffeomorphism, that is, $\Lag$ is a hyperregular Lagrangian function.

The Hamiltonian function $h \in \Cinfty(\Tan^*(\Tan\R^3))$ is
$$
h = (\Leg^{-1})^*E_\Lag = \sum_{A=1}^{3} \left[p_A^0q_1^A - \frac{1}{2} \left( (q_1^A)^2 + (p_A^1)^2 \right)\right].
$$

From this Hamiltonian, applying the procedure given
in Section \eqref{HamHamProb}, one can obtain the Hamilton-Jacobi
equation for this problem which coincides with the Hamilton-Jacobi
equation given previously in the Lagrangian problem.

A particular solution of this Hamilton-Jacobi equation in dimension $1$ 
has been obtained in \cite{art:Constantelos84}. This particular solution is
$$
W(q_0,q_1) = \sqrt{2}\int \d q_{1}\sqrt{-\frac{1}{2}q_{1}^{2}+c_{2}q_{1}-c_{1}}+c_{2}q_{0}
\quad , \quad (c_1, c_2 \in \R) \ .
$$

\subsection{A (homogeneous) deformed elastic cylindrical beam with fixed ends}
\label{sec:Example2}

Consider a deformed elastic cylindrical beam with both
ends fixed. The problem is to determinate its shape; that is, the width of every
section transversal to the axis. This gives rise to a $1$-dimensional
second-order dynamical system, which is autonomous if we require the
beam to be homogeneous \cite{book:Benson06,book:Elsgoltz83,art:Prieto_Roman12_1}.
Let $Q$ be the $1$-dimensional smooth
manifold modeling the configuration space of the system with
local coordinate $(q_0)$. Then, in the natural coordinates of $\Tan^2Q$,
the Lagrangian function for this system is
$$
\Lag(q_0,q_1,q_2) = \frac{1}{2}\mu q_2^2 + \rho q_0 \, ,
$$
where $\mu,\rho \in \R$ are constants, and $\mu\neq0$. This is a
regular Lagrangian function because the Hessian matrix $\ds \left(
\derpars{\Lag}{q_2}{q_2} \right) = \mu$ has maximum rank equal to
$1$ when $\mu \neq 0$.


The local expressions for the Poincar\'{e}-Cartan
forms $\theta_\Lag \in \df^1(\Tan^3Q)$ and $\omega_\Lag\in\df^2(\Tan^3Q)$,
and the Lagrangian energy $E_\Lag \in \Cinfty(\Tan^3Q)$ are
\begin{align*}
\theta_\Lag = \mu ( -q_3 \d q_0 + \mu q_2 \d q_1) \ ; \
\omega_\Lag = \mu (-\d q_0 \wedge \d q_3 + \d q_1 \wedge \d q_2) \ ; \
E_\Lag = -\rho q_0 + \frac{1}{2}\mu q_2^2 - \mu q_1q_3 \, .
\end{align*}
Thus, the semispray of type $1$ $X_\Lag \in \vf(\Tan^3Q)$ solution to the
dynamical equation \eqref{eqn:LagDynEq} is given locally by
$$
X_\Lag = q_1\derpar{}{q_0} + q_2\derpar{}{q_1} + q_3\derpar{}{q_2} - \frac{\rho}{\mu}\derpar{}{q_3} \, .
$$
Observe that the Euler-Lagrange equation for this $1$-dimensional system is
$$
\frac{d^4}{dt^4}\gamma = -\frac{\rho}{\mu} \, ,
$$
where $\gamma \colon \R \to Q$ is a curve. Therefore, it is straightforward
to obtain the general solution, which is a $4$th order polynomial given by
$$
\gamma(t) = -\frac{\rho}{\mu}t^4 + c_3t^3 + c_2 t^2 + c_1 t + c_0
$$
where $c_0,c_1,c_2,c_3 \in \R$ are constants depending on the initial conditions
given.

Now, we state the equations of the Lagrangian Hamilton-Jacobi
problem for this system.
Consider the projection $\rho^3_1\colon\Tan^3Q \to \Tan Q$.
By Proposition \ref{prop:GenLagHJEquiv}, the generalized second-order Lagrangian
Hamilton-Jacobi problem consists in finding sections
$s \in \Gamma(\rho^3_1)$, given locally by
$s(q_0,q_1) = (q_0,q_1,s_2,s_3)$, such that the submanifold
$\textnormal{Im}(s) \hookrightarrow \Tan^3Q$ is invariant by the Lagrangian
vector field $X_\Lag \in \vf(\Tan^3Q)$. Since the constraints defining locally $\textnormal{Im}(s)$
are $q_2-s_2 = 0$, $q_3-s_3=0$, then the equations for the section $s$ are
$$
s_3 - q_1\derpar{s_2}{q_0} - s_2\derpar{s_2}{q_1} = 0 \quad ; \quad
-\frac{\rho}{\mu} - q_1\derpar{s_3}{q_0} - s_2\derpar{s_3}{q_1} = 0 \, .
$$

For the second-order Lagrangian Hamilton-Jacobi problem,
we must require, in addition, that the section $s \in \Gamma(\rho^3_1)$
satisfies $\d (s^*E_\Lag) = s^*\d E_\Lag = 0$.
From the local expression of the Lagrangian energy
$E_\Lag \in \Cinfty(\Tan^3Q)$ given above, we have
$$
\d E_\Lag = -\rho \d q_0 - \mu q_3 \d q_1 + \mu q_2 \d q_2 - \mu q_1 \d q_3 \, .
$$
Thus, taking the pull-back of $\d E_\Lag$ by the section
$s(q_0,q_1) = (q_0,q_1,s_2,s_3)$, we obtain
$$
s^*\d E_\Lag = \mu \left( -\frac{\rho}{\mu} + s_2\derpar{s_2}{q_0} - q_1 \derpar{s_3}{q_0} \right) \d q_0
+ \mu \left( - s_3 + s_2 \derpar{s_2}{q_1} - q_1\derpar{s_3}{q_1} \right) \d q_1 \, .
$$
Hence, the section $s \in \Gamma(\rho^3_1)$ is a solution to the
second-order Lagrangian Hamilton-Jacobi problem if its component functions
satisfy the following system of partial differential equations:
\begin{align*}
&s_3 - q_1\derpar{s_2}{q_0} - s_2\derpar{s_2}{q_1} = 0 \quad ; \quad
-\frac{\rho}{\mu} - q_1\derpar{s_3}{q_0} - s_2\derpar{s_3}{q_1} = 0 \\
&-\frac{\rho}{\mu} - q_1 \derpar{s_3}{q_0} + s_2\derpar{s_2}{q_0} = 0 \quad ; \quad
- s_3 - q_1\derpar{s_3}{q_1} + s_2 \derpar{s_2}{q_1} = 0
\end{align*}
These $4$ partial differential equations are not linearly independent.
In particular, the equations obtained requiring $\d(s^*E_\Lag) = 0$
can be reduced to a single
one by computing the pull-back of the Poincar\'{e}-Cartan $2$-form by the section $s$,
$$
s^*\omega_\Lag = -\mu \left( \derpar{s_3}{q_1} + \derpar{s_2}{q_0} \right) \d q_0 \wedge \d q_1 \, .
$$
Therefore, requiring $s^*\omega_\Lag = 0$ instead of the equivalent condition
$\d(s^*E_\Lag) = 0$, we have that $s$ is a solution to the second-order Lagrangian Hamilton-Jacobi
problem if its component functions satisfy the following
equivalent system of partial differential equations
$$
s_3 - q_1\derpar{s_2}{q_0} - s_2\derpar{s_2}{q_1} = 0 \ ; \
-\frac{\rho}{\mu} - q_1\derpar{s_3}{q_0} - s_2\derpar{s_3}{q_1} = 0 \ ; \
\derpar{s_3}{q_1} + \derpar{s_2}{q_0} = 0 \, ,
$$
where these $3$ equations are now linearly independent.

Finally, we compute the equations for the generating function $W$.
The pull-back of $\theta_\Lag$ by
$s$ gives, in coordinates:
$$
s^*\theta_\Lag = -\mu s_3 \d q_0 + \mu s_2 \d q_1 \, .
$$
Thus, requiring $s^*\theta_\Lag = \d W$, for a local function $W$ in $\Tan Q$, we obtain
$$
\derpar{W}{q_0} = - \mu s_3 \quad ; \quad
\derpar{W}{q_1} = \mu s_2 \, .
$$
and thus from $\d(s^*E_\Lag) = 0$, we have $s^*E_\Lag = \mbox{const.}$, that is,
$$
-\rho q_0 + \frac{1}{2\mu}\left( \derpar{W}{q_1} \right)^2 + q_1\derpar{W}{q_0} = \mbox{const.} \, ,
$$
which is a the Lagrangian Hamilton-Jacobi equation for this problem.

\noindent\textbf{Remark:}
Observe that, in this particular example, the Hamilton-Jacobi equation
is clearly more difficult to solve than the Euler-Lagrange equation.
Therefore, this example shows that it is important to be careful when applying
the Hamilton-Jacobi theory to a system, since the Hamilton-Jacobi equations obtained
can be harder to solve than the usual Euler-Lagrange (or Hamilton's) equations
of the system.
Nevertheless, observe that a solution of the system can be obtained from a solution
$\gamma \colon \R \to Q$ of the Euler-Lagrange equations as (see \cite{Vi-12})
$$
W(q_0,q_1) = \int_{t_0}^{t_1} \Lag (j^2\gamma(t)) \, \d t \ .
$$


Now, to establish the Hamiltonian formalism for the Hamilton-Jacobi
problem, we consider natural coordinates on
$\Tan^*\Tan Q$ and in these coordinates the Legendre-Ostrogradsky
map $\Leg \colon \Tan^3Q \to \Tan^*\Tan Q$ associated to the
Lagrangian function $\Lag$ is locally given by
$$
\Leg^*q_0 = q_0 \quad ; \quad
\Leg^*q_1 = q_1 \quad ; \quad
\Leg^*p^0 = -\mu q_3 \quad ; \quad
\Leg^*p^1 = \mu q_2 \, .
$$
Moreover, the inverse map $\Leg^{-1} \colon \Tan^*\Tan Q \to
\Tan^3Q$ is
$$
(\Leg^{-1})^*q_0 = q_0 \quad ; \quad
(\Leg^{-1})^*q_1 = q_1 \quad ; \quad
(\Leg^{-1})^*q_2 = \frac{p^1}{\mu} \quad ; \quad
(\Leg^{-1})^*q_3 = -\frac{p_0}{\mu} \, .
$$
From these coordinate expressions it is clear that $\Lag$ is a
hyperregular Lagrangian function, since the Legendre-Ostrogradsky
map is a global diffeomorphism.

The Hamiltonian function $h \in \Cinfty(\Tan^*\Tan Q)$ is
$$
h = (\Leg^{-1})^*E_\Lag = -\rho q_0 + \frac{(p^1)^2}{2\mu} + q_1p^0 \, .
$$

From this Hamiltonian applying the procedure given
in Section \eqref{HamHamProb} one can obtain the Hamilton-Jacobi
equations for this problem which coincides with the Hamilton-Jacobi
equations given previously by the Lagrangian problem, which is
$$
-\rho q_0 + \frac{1}{2\mu} \left(\derpar{W}{q_1}\right)^2 + q_1 \derpar{W}{q_0} = \mbox{const.} \,
$$

\section*{Acknowledgments}

We acknowledge the financial support of the \textsl{Ministerio de Ciencia e Innovaci\'{o}n} (Spain),
projects MTM 2010-21186-C02-01, MTM2011-22585 and  MTM2011-15725-E;
AGAUR, project 2009 SGR:1338.; IRSES-project ``Geomech-246981'';
and ICMAT Severo Ochoa project SEV-2011-0087.
P.D. Prieto-Mart\'{\i}nez wants to thank the UPC for a Ph.D grant, and
L.Colombo wants to thank CSIC for a JAE-Pre grant.
We wish to thanks Profs. D. Mart\'{\i}n de Diego and M.C. Mu\~{n}oz-Lecanda for fruitful discussions.
We thank Mr. Jeff Palmer for his assistance in preparing the English version of the manuscript.


\end{document}